\documentclass[a4paper,UKenglish,cleveref, autoref, thm-restate]{lipics-v2021}

\title{Gradual Program Analysis for Null Pointers}

\author
  {Sam Estep}
  {Carnegie Mellon University, Pittsburgh, PA, USA}
  {}{}{}
\author
  {Jenna Wise}
  {Carnegie Mellon University, Pittsburgh, PA, USA}
  {}{}{}
\author
  {Jonathan Aldrich}
  {Carnegie Mellon University, Pittsburgh, PA, USA}
  {}{}{}
\author
  {Éric Tanter}
  {Computer Science Department (DCC), University of Chile, Santiago, Chile}
  {}{}{FONDECYT Regular project 1190058}
\author
  {Johannes Bader}
  {Jane Street, New York, NY, USA}
  {}{}{}
\author
  {Joshua Sunshine}
  {Carnegie Mellon University, Pittsburgh, PA, USA}
  {}{}{}

\authorrunning{S. Estep, J. Wise, J. Aldrich, É. Tanter, J. Bader, and J. Sunshine}

\Copyright{Sam Estep, Jenna Wise, Jonathan Aldrich, Éric Tanter, Johannes Bader, and Joshua Sunshine}



\supplement{\url{https://github.com/orgs/gradual-verification/packages/container/package/ecoop21}}

\funding{National Science Foundation under Grant No. CCF-1901033 and Grant No. CCF-1852260}

\nolinenumbers

\EventEditors{Manu Sridharan and Anders M{\o}ller}
\EventNoEds{2}
\EventLongTitle{35th European Conference on Object-Oriented Programming (ECOOP 2021)}
\EventShortTitle{ECOOP 2021}
\EventAcronym{ECOOP}
\EventYear{2021}
\EventDate{July 12--16, 2021}
\EventLocation{Aarhus, Denmark (Virtual Conference)}
\EventLogo{}
\SeriesVolume{194}
\ArticleNo{4}

\usepackage{algorithm}        
\usepackage{algpseudocode}    
\usepackage{savesym}          
\savesymbol{Bbbk}
\usepackage{amssymb}          
\usepackage{booktabs}         
\usepackage[outline]{contour} 
\usepackage{listings}         
\usepackage{mathtools}        
\usepackage{mdframed}         
\usepackage{microtype}        
\usepackage{physics}          
\usepackage{plstx}            
\usepackage{stmaryrd}         
\usepackage{subcaption}       
\usepackage{tikz}             
\usepackage{pgfplots}         
\savesymbol{widering}
\usepackage{yhmath}           
\usepackage{mathpartir}       

\contourlength{0.2em}

\definecolor{javared}{rgb}{0.6,0,0} 
\definecolor{javagreen}{rgb}{0.25,0.5,0.35} 
\definecolor{javapurple}{rgb}{0.5,0,0.35} 
\definecolor{javadocblue}{rgb}{0.25,0.35,0.75} 
\definecolor{lightblue}{rgb}{0.68,0.84,0.90}

\newmdenv[
  hidealllines=true,
  innertopmargin=2ex,
  innerbottommargin=2ex,
  innerleftmargin=2em,
]{runningexample*}

\DisableLigatures[-]{family=tt*}

\usetikzlibrary{
  arrows,
  arrows.meta,
  graphs,
  graphs.standard,
  positioning,
  quotes
}

\DeclareMathOperator{\dom}{dom}

\DeclarePairedDelimiter{\bb}{\llbracket}{\rrbracket}

\DeclarePairedDelimiterX{\fram}[2]{\langle}{\rangle}{#1, #2}
\DeclarePairedDelimiterX{\set}[2]{\{}{\}}{#1 : #2}
\DeclarePairedDelimiterX{\stat}[2]{\langle}{\rangle}{#1 \mathrel{\delimsize\|} #2}

\newcommand{\etal}{\emph{et al.}}
\newcommand{\ie}{\emph{i.e.}~}
\newcommand{\eg}{\emph{e.g.}~}

\newcommand*{\func}{\textsc}
\newcommand*{\fAnd}{\func{and}}

\newcommand*{\fBranch}{\func{branch}}

\newcommand*{\fConc}{\func{conc}}

\newcommand*{\fDescend}{\func{descend}}
\newcommand*{\fFlow}{\func{flow}}
\newcommand*{\fInst}{\func{inst}}
\newcommand*{\fKildall}{\func{Kildall}}

\newcommand*{\fNew}{\func{new}}

\newcommand*{\fOr}{\func{or}}
\newcommand*{\fProc}{\func{proc}}
\newcommand*{\fSafe}{\func{safe}}

\newcommand*{\predd}{\textsc}
\newcommand*{\pDesc}{\predd{desc}}

\newcommand*{\sett}{\textsc}
\newcommand*{\sAbst}{\sett{Abst}}
\newcommand*{\sAnn}{\sett{Ann}}
\newcommand*{\sArc}{\sett{Edge}}

\newcommand*{\sEnv}{\sett{Env}}
\newcommand*{\sFrame}{\sett{Frame}}

\newcommand*{\sInst}{\sett{Inst}}
\newcommand*{\sMap}{\sett{Map}}
\newcommand*{\sMem}{\sett{Mem}}

\newcommand*{\sProc}{\sett{Proc}}
\newcommand*{\sProg}{\sett{Prog}}
\newcommand*{\sStack}{\sett{Stack}}
\newcommand*{\sState}{\sett{State}}

\newcommand*{\sVal}{\sett{Val}}
\newcommand*{\sVar}{\sett{Var}}
\newcommand*{\sExpr}{\sett{Expr}}
\newcommand*{\sStmt}{\sett{Stmt}}

\newcommand*{\sField}{\sett{Field}}
\newcommand*{\sVert}{\sett{Vert}}

\newcommand*{\code}{\texttt}

\newcommand*{\vnull}{\code{null}}
\newcommand*{\stmtSkip}{\code{skip}}
\newcommand*{\stmtSeq}[2]{#1 ~;~ #2}
\newcommand{\stmtDecl}[2]{#1~#2}
\newcommand*{\stmtAssign}[2]{#1 ~\code{:=}~ #2}
\newcommand*{\stmtBranch}[1]{\code{branch}~#1}
\newcommand*{\stmtIf}[1]{\code{if}~#1}
\newcommand*{\stmtElse}[1]{\code{else}~#1}
\newcommand*{\stmtIfElse}[3]{\code{if}~(#1)~\{~#2~\}~\code{else}~\{~#3~\}}
\newcommand{\stmtWhile}[2]{\code{while}~(#1)~\{~#2~\}}
\newcommand*{\stmtCall}[3]{#1~\code{:=}~#2(#3)}
\newcommand*{\stmtProc}[2]{\code{proc}~#1(#2)}
\newcommand*{\exprCall}[2]{#1(#2)}
\newcommand*{\stmtReturn}[1]{\code{return}~#1}
\newcommand{\stmtFieldAssign}[3]{#1.#2~\code{:=}~#3}

\newcommand{\stmtMain}{\code{main}}
\newcommand{\exprNew}[1]{\code{new}(#1)}

\newcommand*{\cProc}{\code{proc}}

\newcommand*{\cMain}{\code{main}}
\newcommand*{\cError}{\code{error}}

\newcommand*{\ann}{\code}
\newcommand*{\aNa}{\ann{Nullable}}
\newcommand*{\aNu}{\ann{Null}}
\newcommand*{\aNN}{\ann{NonNull}}

\newcommand*{\Intersect}{\bigcap}
\newcommand*{\Nat}{\mathbb{N}}

\newcommand*{\arcsto}[3]{#2 \xrightarrow{#1} #3}
\newcommand*{\bigjoin}{\bigsqcup}
\newcommand*{\inst}[2][]{[#2]_{#1}}
\newcommand*{\intersect}{\cap}
\newcommand*{\join}{\sqcup}

\newcommand*{\nil}{\mathsf{nil}}
\newcommand*{\preciser}{\lesssim}

\newcommand*{\powerset}{\mathcal{P}}
\newcommand*{\pto}{\rightharpoonup}
\newcommand*{\qm}[1][]{#1\texttt{?}}
\newcommand*{\stepsto}{\longrightarrow}
\newcommand*{\strictunder}{\sqsubset}
\newcommand*{\under}{\sqsubseteq}
\newcommand*{\union}{\cup}

\newcommand*{\Coll}[1]{\wideparen{#1}}
\newcommand*{\Grad}[1]{\widetilde{#1}}
\newcommand*{\Gradjoin}{\mathbin{\Grad{\join}}}
\newcommand*{\Gradstepsto}{\mathrel{\Grad{\stepsto}}}
\newcommand*{\Gradunder}{\mathrel{\Grad{\under}}}

\begin{document}

\maketitle

\begin{abstract}
Static analysis tools typically address the problem of excessive false positives by requiring programmers to explicitly annotate their code. However, when faced with incomplete annotations, many analysis tools are either too conservative, yielding false positives, or too optimistic, resulting in unsound analysis results. In order to flexibly and soundly deal with partially-annotated programs, we propose to build upon and adapt the gradual typing approach to abstract-interpretation-based program analyses. Specifically, we focus on null-pointer analysis and demonstrate that a gradual null-pointer analysis hits a sweet spot, by gracefully applying static analysis where possible and relying on dynamic checks where necessary for soundness.
In addition to formalizing a gradual null-pointer analysis for a core imperative language, we build a prototype using the Infer static analysis framework, and
present preliminary evidence that the gradual null-pointer analysis reduces false positives compared to two existing null-pointer checkers for Infer.
Further, we discuss ways in which the gradualization approach used to derive the gradual analysis from its static counterpart can be extended to support more domains. This work thus provides a basis for future analysis tools that can smoothly navigate the tradeoff between human effort and run-time overhead to reduce the number of reported false positives.
\end{abstract}

\section{Introduction}

Static analysis is useful \cite{ayewah2010google}, but underused in practice because of false positives \cite{Johnson2013}. A commonly-used way to reduce false positives is through programmer-provided annotations~\cite{barnett2007annotations} that make programmers intent manifest. For example, Facebook's Infer Eradicate~\cite{FacebookInfer}, Uber's \textsc{NullAway}~\cite{banerjee2019nullaway}, and the Java Nullness Checker from the Checker Framework~\cite{papi2008practical} all rely on \code{@NonNull} and \code{@Nullable} annotations to statically find and report potential null-pointer exceptions in Java code. However, in practice, annotating code completely can be very costly~\cite{chalin2007non}---or even impossible, for instance, when relying on third-party libraries and APIs.
As a result, since non-null reference variables are used extensively in software~\cite{chalin2007non}, many tools assume missing annotations are \code{@NonNull}. But, the huge number of false positives produced by such an approach in practice is a serious burden. To address this pitfall, \textsc{NullAway} assumes that sinks (i.e. targets of assignments and bindings) are \code{@Nullable} and sources are \code{@NonNull}. Unfortunately, both strategies are unsound, and therefore programs deemed valid may still raise null pointer exceptions at run time.

This paper explores a novel approach to these issues by drawing on research in gradual typing~\cite{siek2006gradual,siek2015refined,garcia2016abstracting} and its recent adaptation to gradual verification~\cite{bader2018gradual, wise2020gradual}.
We propose gradual program analysis as a principled, sound, and practical way to handle missing annotations. As a first step in the research agenda of gradual program analysis, this article studies the case of a simple null-pointer analysis. We present a general formal framework to derive gradual program analyses by transforming static analyses based on abstract interpretation~\cite{cousot:popl1977}. Specifically, we study analyses that operate over first-order procedural imperative languages and support user-provided annotations. This setting matches the core language used by many tools, such as Infer.
In essence, a \textit{gradual analysis} treats missing annotations optimistically, but injects run-time checks to preserve soundness. Crucially, the static portion of a gradual analysis uses the same algorithmic architecture as the underlying static analysis.\footnote{Note that an alternative is phrasing nullness as a type system, which can also be gradualized~\cite{brotherston2017granullar,blameForNull}.  We focus on approaches based on static analysis, which have very different technical foundations and user experience.  We compare to type-based approaches in Section~\ref{sec:related}.}

Additionally, we ensure that any gradual analysis produced from our framework satisfies the \textit{gradual guarantees}, adapted from Siek \etal~\cite{siek2015refined} formulation for gradual typing. Any gradual analysis is also a \textit{conservative extension} of the base static analysis: when all annotations are provided, the gradual analysis is equivalent to the base static analysis, and no run-time checks are inserted.
Therefore, the gradual analysis smoothly trades off between static and dynamic checking, driven by the annotation effort developers are willing to invest.

To provide initial evidence of the applicability of gradual null-pointer analysis, we implement a gradual null-pointer analysis (GNPA) using Facebook's Infer analysis framework and report on preliminary experiments using the prototype.\footnote{\url{https://github.com/orgs/gradual-verification/packages/container/package/ecoop21}} The experiments show that a gradual null-pointer analysis can be effectively implemented, and used at scale to produce a small number of false positives in practice---fewer than Infer \textsc{Eradicate} as well as a more recent Infer checker, \textsc{NullSafe}. They also show that GNPA eliminates on average more than half of the null-pointer checks Java automatically inserts at run time. As a result, unlike other null-pointer analyses, GNPA can both prove the redundancy of run-time checks and reduce reported false positives.

The rest of the paper is organized as follows. In Section~\ref{sec:action}, we motivate gradual program analysis in the setting of null pointers by looking at how \textsc{Eradicate}, \textsc{NullSafe}, \textsc{NullAway}, and the Java Nullness Checker operate on example code with missing annotations, showcasing the concrete advantages of GNPA.
Section \ref{sec:language} formalizes PICL, a core imperative language similar to that of Infer. Section \ref{sec:static} then presents the static null-pointer analysis (NPA) for PICL, which is then used as the starting point for the derivation of the gradual analysis. We describe our approach to gradualizing a static program analysis in Section \ref{sec:gradual}, using GNPA as the running case study. Additionally, Section~\ref{sec:gradual} includes a discussion of important gradual properties our analysis adheres to: \textit{soundness}, \textit{conservative extension}, and the \textit{gradual guarantee}. All proofs can be found in the appendix of the full version of this paper.
We report on the preliminary empirical evaluation of an Infer GNPA checker called \emph{Graduator} in Section \ref{sec:empirical}.
Section~\ref{sec:related} discusses
related work and Section~\ref{sec:conclusion} concludes. In the conclusion, we sketch ways in which the approach presented here could be applied to other analysis domains, highlight open venues for future work in the area of gradual program analysis.

\section{Gradual Null-Pointer Analysis in Action}\label{sec:action}

This section informally introduces gradual null-pointer analysis and its potential compared to existing approaches through a simple example. We first briefly recall the basics of null-pointer analyses, and then discuss how current tools deal with missing annotations in problematic ways.

\subsection{Null-Pointer Analysis in a Nutshell}
\label{sec:npa-intro}

With programming languages that allow any type to be inhabited by a null value, programmers end up facing runtime errors (or worse if the language is unsafe) related to dereferencing null pointers. A null-pointer analysis is a static analysis that detects \emph{potential} null pointer dereferences and reports them as warnings, so that programmers can understand where explicit nullness checks should be added in order to avoid runtime errors.
Examples of null-pointer analyses are Infer Eradicate \cite{FacebookInferEradicate} and the Java Null Checker \cite{papi2008practical}.
Typically, a null-pointer analysis allows programmers to add annotations in the code to denote which variables (as well as fields, return values, etc.) are, or should be, non-null--e.g. \code{@NonNull}--and which are potentially null--e.g. \code{@Nullable}. A simple flow analysis is then able to detect and report conflicts, such as when a nullable variable is assigned to a non-null variable.

While a static null pointer analysis brings guarantees of robustness to a codebase, its adoption is not necessarily seamless. If a static analysis aims to be sound, it
must not suffer from false negatives, i.e. miss any actual null pointer dereference that can happen at runtime. While desirable, this means the analysis necessarily has to be conservative and therefore reports false positives---locations that are thought to potentially trigger null pointer dereferences, but actually do not.

This standard static analysis conundrum is exacerbated when considering programs where not all variables are annotated. Of course, in practice, a codebase is rarely fully annotated. Existing null-pointer analyses assign missing annotations a concrete annotation, such as $\aNa$ or $\aNN$. In doing so, they either report additional false positives, suffer from false negatives (and hence are unsound), or both. The rest of this section illustrates these issues with a simple example, and discusses how a gradual null-pointer analysis (GNPA) alleviates them. GNPA treats missing annotations in a special manner, following the gradual typing motto of being optimistic statically and relying on runtime checks for soundness~\cite{siek2006gradual}. Doing so allows the analysis to leverage both static and dynamic information to reduce false positives while maintaining soundness.

\subsection{Avoiding False Positives} \label{sec:action-falsepos}
GNPA can reduce the number of false positives reported by static tools by leveraging provided annotations and run-time checks. We demonstrate this with the unannotated program in Figure~\ref{fig:java-safereverse}. The program appends the reverse of a non-null string to the reverse of a null string and prints the result. The \code{reverse} method (lines \ref{line:safereverse-start}--\ref{line:safereverse-end}) returns the reverse of an input string when it is non-null and an empty string when the input is \code{null}. Additionally, \code{reverse} is unannotated, as highlighted for reference.

\begin{figure}
  \begin{center}
    \begin{minipage}{0.8\linewidth}
      \begin{lstlisting}[language=Java,numbers=left,escapechar=|]
class Main {

  static |\colorbox{lightblue}{\phantom{@?}}| String reverse(|\colorbox{lightblue}{\phantom{@?}}| String str) {|\label{line:safereverse-start}|
    if (str == null) return new String();|\label{line:return-nonnull}|
    StringBuilder builder = new StringBuilder(str);
    builder.reverse();
    return builder.toString();
  }|\label{line:safereverse-end}|

  public static void main(String[] args) {
    String reversed = reverse(null);|\label{line:safereverse-null}|
    String frown = reverse(":)");|\label{line:safereverse-smile}|
    String both = reversed.concat(frown);|\label{line:safereverse-concat}|
    System.out.println(both);
  }
}
\end{lstlisting}
    \end{minipage}
  \end{center}
  \caption{Unannotated Java code safely reversing nullable strings.}
  \label{fig:java-safereverse}
\end{figure}

The most straightforward approach to handling the missing annotations is to replace them with a fixed annotation.  Infer Eradicate and the Java Nullness Checker both choose \code{@NonNull} as the default, since that is the most frequent annotation used in practice~\cite{chalin2007non}.  Thus, in this example, they would treat \code{reverse}'s argument and return value as annotated with \code{@NonNull}. This correctly assigns \code{reversed} and \code{frown} as non-null on lines \ref{line:safereverse-null} and \ref{line:safereverse-smile}; and consequently, no false positive is reported when \code{reversed} is dereferenced on line \ref{line:safereverse-concat}. However, both tools will report a false positive each time \code{reverse} is called with \code{null}, as in line \ref{line:safereverse-null}.

Other uniform defaults are possible, but likewise lead to false positive warnings.  For example, choosing \code{@Nullable} by default would result in a false positive when \code{reversed} is dereferenced.  A more sophisticated choice would be the Java Nullness Checker's \code{@PolyNull} annotation, which supports type qualifier polymorphism for methods annotated with \code{@PolyNull}.  If \code{reverse}'s method signature is annotated with \code{@PolyNull}, then \code{reverse} would have two conceptual versions:
\begin{center}
\code{static @Nullable String reverse(@Nullable String str)}
\\
\code{static @NonNull String reverse(@NonNull String str)}
\end{center}
At a call site, the most precise applicable signature would be chosen; so, calling \code{reverse} with \code{null} (line \ref{line:safereverse-null}) would result in the \code{@Nullable} signature, and calling \code{reverse} with \code{":)"} (line \ref{line:safereverse-smile}) would result in the \code{@NonNull} signature. Unfortunately, this strategy marks \code{reversed} on line \ref{line:safereverse-null} as \code{@Nullable} even though it is \code{@NonNull}, and a false positive is reported when \code{reversed} is dereferenced on line \ref{line:safereverse-concat}.  So while \code{@PolyNull} increases the expressiveness of the annotation system, it does not solve the problem of avoiding false positives from uniform annotation defaults.

In contrast, GNPA optimistically assumes both calls to \code{reverse} in \code{main} (lines \ref{line:safereverse-null}--\ref{line:safereverse-smile}) are valid without assigning fixed annotations to \code{reverse}'s argument or return value. Then, the analysis can continue relying on \textit{contextual optimism} when reasoning about the rest of \code{main}: \code{reversed} is assumed \code{@NonNull} to satisfy its dereference on line \ref{line:safereverse-concat}.
Of course this is generally an unsound assumption, so a run-time check is inserted to ascertain the non-nullness of \code{reversed} and preserve soundness.
Alternatively, a developer could annotate the return value of \code{reverse} with \code{@NonNull}. GNPA will operate as before except it will leverage this new information during static reasoning. Therefore, \code{reversed} will be marked \code{@NonNull} on line \ref{line:safereverse-null} and the dereference of \code{reversed} on line \ref{line:safereverse-concat} will be statically proven safe without any run-time check.

It turns out that a non-uniform choice of defaults can be optimistic in the same sense as GNPA.  For example, \textsc{NullAway} assumes sinks are \code{@Nullable} and sources are \code{@NonNull} when annotations are missing. In fact, this strategy correctly annotates \code{reverse}, and so no false positives are reported by the tool for the program in Figure \ref{fig:java-safereverse}.
However, in contrast to the gradual approach, the \textsc{NullAway} approach is in fact unsound, as illustrated next.

\subsection{Avoiding False Negatives} \label{sec:action-falseneg}
When Eradicate, \textsc{NullAway}, and the Java Nullness Checker handle missing annotations, they all give up soundness in an attempt to limit the number of false positives produced.

To illustrate, consider the same program from Figure~\ref{fig:java-safereverse}, with one single change: the \code{reverse} method now returns \code{null} instead of an empty string (line \ref{line:return-nonnull}).
\begin{lstlisting}[language=Java,numbers=none,escapechar=|]
   if (str == null) return null;
\end{lstlisting}
All of the tools mentioned earlier, including \textsc{NullAway}, erroneously assume that the return value of \code{reverse} is \code{@NonNull}. On line \ref{line:safereverse-null}, \code{reversed} is assigned \code{reverse(null)}'s return value of \code{null}; so, it is an error to dereference \code{reversed} on line \ref{line:safereverse-concat}. Unfortunately, all of the tools assume \code{reversed} is assigned a non-null value and do not report an error on line \ref{line:safereverse-concat}. This is a \textit{false negative}, which means that at runtime the program will fail with a null-pointer exception.

GNPA is similarly optimistic about \code{reversed} being non-null on line \ref{line:safereverse-concat}. However, GNPA safeguards its optimistic static assumptions with run-time checks. Therefore, the analysis will correctly report an error on line \ref{line:safereverse-concat}.
Alternatively, a developer could annotate the return value of \code{reverse} with \code{@Nullable}. By doing so, the gradual analysis will be able to exploit this information statically to report a static error, instead of a dynamic error.

\mbox{}

To sum up, a gradual null-pointer analysis can reduce false positives by optimistically treating missing annotations, and preserve soundness by detecting errors at runtime. Of course, one may wonder why it is better to fail at runtime when passing a null value as a non-null annotated argument, instead of just relying on the upcoming null-pointer exception. There are two answers to this question. First, in unsafe languages like C, a null-pointer dereference results in a crash. Second, in a safe language like Java where a null-pointer dereference is anyway detected and reported, it can be preferable to fail as soon as possible, in order to avoid performing computation (and side effects) under an incorrect assumption. This is similar to how the eager reporting of gradual typing can be seen as an improvement over simply relying on the underlying safety checks of a dynamically-typed language.

\mbox{}

Next, we formally develop GNPA, prove that it is sound, and prove that it smoothly trades-off between static and dynamic checking following the gradual guarantee criterion from gradual typing~\cite{siek2015refined}. We finally report on an actual implementation of GNPA and compare its effectiveness with existing tools.

\section{PICL: A Procedural Imperative Core Language} \label{sec:language}

Following the Abstract Gradual Typing methodology introduced by Garcia \etal~\cite{garcia2016abstracting}, we build GNPA on top of a static null-pointer analysis, NPA. Thus, we first formally present a procedural imperative core language (PICL), used for both analyses to operate on; we present NPA in Section \ref{sec:static}, and GNPA in Section \ref{sec:gradual}.
PICL is akin to the intermediate language of the Infer framework, and therefore the formal development around PICL drove the implementation of the Infer GNPA checker we evaluate in Section~\ref{sec:empirical}.

\begin{figure}[]
  \centering
 \begin{minipage}[t]{0.47\linewidth}
  \begin{plstx}
    *: x, y [\in] \sVar \\
    *: e [\in] \sExpr \\
    *: a [\in] \sAnn = \{\code{Nullable},~\code{NonNull},~\qm\}\\
    : P ::= \overline{procedure} ~\overline{field} ~s \\
    : field ::= T ~f; \\
    : procedure ::= T\code{@}a~ m ~(~\overline{T\code{@}a ~x}~)~ \{ ~s~ \} \\
    : T ::= \code{ref} \\
    : \oplus ::= \land | \lor \\
  \end{plstx}
\end{minipage}\hfill
 \begin{minipage}[t]{0.47\linewidth}
  \begin{plstx}
    *: m [\in] \sProc \\
    *: f [\in] \sField \\
    *: s [\in] \sStmt \\
    : e ::= \vnull | x | e \oplus e | e.f | \exprNew{\overline{f}} | \exprCall{m}{x} \\
    : c ::= e = \vnull | e \neq \vnull \\
    : s ::= \stmtSkip
    | \stmtSeq{s}{s}
    | \stmtDecl{T}{x}
    | \stmtAssign{x}{e}
    | \stmtFieldAssign{x}{f}{y}
    | \stmtIfElse{c}{s}{s}
    | \stmtWhile{c}{s}
    | \stmtReturn{y} \\
  \end{plstx}
\end{minipage}
  \caption{Abstract syntax of PICL.}
  \label{fig:syntax}
\end{figure}

\subsection{Syntax \& Static Semantics}

The syntax of PICL can be found in Figure \ref{fig:syntax}. Programs consist of procedures\footnote{Procedures accept only one parameter to simplify later formalisms.}, fields, and statements. Statements include the empty statement, sequences, variable declarations, variable and field assignments, conditionals, while loops, and returns. Expressions consist of \code{null} literals, variables, comparisons, conjunctions, disjunctions, field accesses, object allocations, and procedure calls. Finally, procedures may have \code{Nullable} or \code{NonNull} annotations on their arguments and return values. Missing annotations are represented by \qm.

As the focus of this work is not on typing, we only consider well-formed and well-typed programs, which is standard and not formalized here. In particular, variables are declared and initialized before use, and field and procedure names are unique.

\subsection{Control Flow Graph Representation}
Well-formed programs written in the abstract syntax given in Fig. \ref{fig:syntax} are translated into \emph{control flow graphs}---one graph for each procedure body and one for the main $s$. A finite control flow graph (CFG) for program $p$ has vertices $\sVert_p$ and edges $\sArc_p \subseteq \sVert_p \times \sVert_p$. For $v_1, v_2 \in \sVert_p$, we write $\arcsto{p}{v_1}{v_2}$ to denote $(v_1, v_2) \in \sArc_p$. Each vertex holds a single instruction, which we can access using the function $\fInst_p : \sVert_p \to \sInst$. We write $\inst[v]{\iota}$ to denote a vertex $v \in \sVert_p$ such that $\fInst_p(v) = \iota$, or just $\inst{\iota}$ (omitting the $v$) when the vertex itself is not important. By construction, these translated CFGs satisfy certain well-formedness properties, listed in the appendix of the full version of this paper.

\begin{figure}[]
  \centering
  \begin{minipage}[t]{0.47\linewidth}
  \begin{plstx}
    *: x, y, z [\in] \sVar \\
    *: a, b [\in] \sAnn = \{\code{Nullable},~\code{NonNull},~\qm\} \\
  \end{plstx}
  \end{minipage}\hfill
  \begin{minipage}[t]{0.47\linewidth}
  \begin{plstx}
    *: m [\in] \sProc \\
    *: f [\in] \sField \\
  \end{plstx}
  \end{minipage}
  \begin{plstx}
    : I ::= \stmtAssign{x}{y}
    | \stmtAssign{x}{\vnull}
    | \stmtCall{x}{m\code{@}a}{y\code{@}b}
    | \stmtAssign{x}{\exprNew{\overline{f}}}
    | \stmtAssign{x}{y \land z}
    | \stmtAssign{x}{y \lor z}
    | \stmtAssign{x}{y.f}
    | \stmtFieldAssign{x}{f}{y}
    | \stmtBranch{x}
    | \stmtIf{x}
    | \stmtElse{x}
    | \stmtReturn{y\code{@}a}
    | \stmtMain
    | \stmtProc{m\code{@}a}{y\code{@}b}\\
  \end{plstx}
  \caption{Abstract syntax of a CFG instruction.}
  \label{fig:instr-syntax}
\end{figure}

The set of possible instructions is defined in Figure~\ref{fig:instr-syntax}. In general, the CFG instructions are atomic variants of program statements designed to simplify the analysis presentations. Figure \ref{fig:ex-cfg} gives the CFG of a simple procedure \code{foo}, which calls \code{bar} repeatedly until \code{x} becomes non-null and then returns \code{x}. The CFG starts with \code{foo}'s entry node $\stmtProc{foo@NonNull}{x@Nullable}$ (similarly, $\stmtMain$ is always the entry node of the main program's CFG). Then, the while loop on lines \ref{ex:cfg-while-start}--\ref{ex:cfg-while-end} results in the $\stmtBranch{x}$ sub-graph, which leads to $\stmtIf{x}$ when $x$ is non-null and $\stmtElse{x}$ when $x$ is null. The call to \code{bar} follows from $\stmtElse{x}$ and loops back to $\stmtBranch{x}$ as expected. Finally, $\stmtReturn{x@NonNull}$ follows from $\stmtIf{x}$ ending the CFG. Precise semantics for instructions is given in Section \ref{sec:dyn-lang-semantics}.

\begin{figure}[]
  \begin{minipage}[t]{0.5\textwidth}
  \centering
  \small
  \begin{lstlisting}[language=Java,numbers=left,escapechar=|]
  ref@NonNull foo(ref@Nullable x)
  {
    |\label{ex:cfg-while-start}|while (x == null)
    {
      x := bar(x);
    |\label{ex:cfg-while-end}|}
    return x;
  }
  \end{lstlisting}
  \end{minipage}\hfill
  \begin{minipage}[t]{0.5\textwidth}
  \strut\vspace*{-\baselineskip}\newline
  \centering
  \begin{tikzpicture}[->]
    \node (v1) {\stmtProc{\code{foo@NonNull}}{\code{x@Nullable}}};
    \node (v2) [below = 0.5cm of v1] {\stmtBranch{\code{x}}};
    \node (v3) [below right = 0.5cm and 0.3cm of v2] {\stmtIf{\code{x}}};
    \node (v4) [below left = 0.5cm and 0.2cm of v2] {\stmtElse{\code{x}}};
    \node (v5) [below right = 0.5cm and -1.5cm of v4] {\stmtCall{\code{x}}{\code{bar@?}}{\code{x@?}}};
    \node (v6) [below right = 0.5cm and -1cm of v3] {\stmtReturn{\code{x@NonNull}}};

    \draw (v1) -> (v2);
    \draw (v2) -> (v3);
    \draw (v2) -> (v4);
    \draw (v4) -> (v5);
    \draw (v5) -> (v2);
    \draw (v3) -> (v6);
  \end{tikzpicture}
  \end{minipage}
  \caption{Example CFG.}
  \label{fig:ex-cfg}
\end{figure}

\subsection{Dynamic Semantics} \label{sec:dyn-lang-semantics}
We define the set of possible object locations as the set of natural numbers and 0, $\sVal = \Nat \cup \{0\}$. The \code{null} pointer is location 0.

Now, a program state ($\sState_p \subseteq \sStack_p \times \sMem_p$) consists of a stack and a heap. A heap $\mu \in \sMem_p = (\sVal \setminus \{0\}) \pto (\sField \pto \sVal)$ maps object locations and field names to program values---other (possibly null) pointers. A stack is made of stack frames each containing a local variable environment and CFG node:
\begin{align*}
  S \in \sStack_p ::= E \cdot S \mid \nil
  \qq{where} & E \in \sFrame_p = \sEnv \times \sVert_p \\
  \qand      & \sEnv = \sVar \pto \sVal.
\end{align*}

Further, we restrict the set of states $\xi = \stat{\fram{\rho_1}{v_1} \cdot \fram{\rho_2}{v_2} \cdots \fram{\rho_n}{v_n} \cdot \nil}{\mu} \in \sState_p$ to include only those satisfying the following conditions:
\begin{enumerate}
  \item \textit{Bottom stack frame is in \cMain:} Let $\fDescend : \sVert_p \to \powerset^{+}(\sVert_p)$ give the descendants of each node in the control flow graph. Then $v_i \in \fDescend(v_0)$ if and only if $i = n$.
  \item \textit{Every variable defaults to \vnull{} (except on \cMain{} and \cProc{} nodes):} If $\fInst_p(v_i) \neq \stmtMain$ and $\fInst_p(v_i) \neq \stmtProc{m@a}{y@b}$ then $\rho_i$ is a total function.
  \item \textit{Follow the ``true'' branch when non-null:} If $\fInst_p(v_i) = \stmtIf{y}$ then $\rho_i(y) \neq 0$.
  \item \textit{Follow the ``false'' branch when null:} If $\fInst_p(v_i) = \stmtElse{y}$ then $\rho_i(y) = 0$.
  \item\label{rule:call} \textit{Every frame except the top is a procedure call:}
    If $v_i \in \fDescend(\stmtProc{m\code{@}a}{y\code{@}b})$ then $\fInst_p(v_{i +
    1}) = \stmtCall{x}{m\code{@}a}{y'\code{@}b}$, and either $b = \qm$ or $\rho_{i + 1}(y') \in \fConc(b)$ (see section~\ref{sec:static}.
\end{enumerate}

Now, the small-step semantics of PICL is given in
Figure~\ref{fig:semantics}, where $\rho_0 = \set{x \mapsto 0}{x \in \sVar}$.
\begin{figure}[]
  \centering
  \begin{footnotesize}
    \begin{gather*}
      \stat{\fram{\rho}{\inst[u]{\stmtAssign{x}{y}}} \cdot S}{\mu} \stepsto_p \stat{\fram{\rho[x \mapsto \rho(y)]}{v} \cdot S}{\mu} \\
      \stat{\fram{\rho}{\inst[u]{\stmtBranch{y}}} \cdot S}{\mu} \stepsto_p \stat{\fram{\rho}{\inst[v]{\fBranch(\rho(y), y)}} \cdot S}{\mu} \\
      \stat{\fram{\rho}{\inst[u]{\stmtIf{y}}} \cdot S}{\mu} \stepsto_p \stat{\fram{\rho}{v} \cdot S}{\mu} \\
      \stat{\fram{\rho}{\inst[u]{\stmtElse{y}}} \cdot S}{\mu} \stepsto_p \stat{\fram{\rho}{v} \cdot S}{\mu} \\
      \stat{\fram{\rho}{\inst[u]{\stmtCall{x}{m\code{@}a}{y\code{@}b}}} \cdot S}{\mu} \stepsto_p \stat{\fram{\varnothing}{\inst{\stmtProc{m\code{@}a}{y'\code{@}b}}} \cdot \fram{\rho}{u} \cdot S}{\mu} \\
      \stat{\fram{\rho_1}{\inst[u]{\stmtProc{m\code{@}a}{y\code{@}b}}} \cdot \fram{\rho_2}{\inst[w]{\stmtCall{x}{m\code{@}a}{y'\code{@}b}}} \cdot S}{\mu} \stepsto_p \stat{\fram{\rho_0[y \mapsto \rho_2(y')]}{v} \cdot \fram{\rho_2}{w} \cdot S}{\mu} \\
      \stat{\fram{\rho_1}{\inst{\stmtReturn{y\code{@}a}}} \cdot \fram{\rho_2}{\inst[u]{\stmtCall{x}{m\code{@}a}{y'\code{@}b}}} \cdot S}{\mu} \stepsto_p \stat{\fram{\rho_2[x \mapsto \rho_1(y)]}{v} \cdot S}{\mu}\;\dagger \\
      \stat{\fram{\rho}{\inst[u]{\stmtAssign{x}{\vnull}}} \cdot S}{\mu} \stepsto_p \stat{\fram{\rho[x \mapsto 0]}{v} \cdot S}{\mu} \\
      \stat{\fram{\rho}{\inst[u]{\stmtAssign{x}{\exprNew{\overline{f}}}}} \cdot S}{\mu} \stepsto_p \stat{\fram{\rho[x \mapsto \fNew(\mu)]}{v} \cdot S}{\mu[\fNew(\mu) \mapsto \overline{[f_i \mapsto \vnull]}]} \\
      \stat{\fram{\rho}{\inst[u]{\stmtAssign{x}{y \land z}}} \cdot S}{\mu} \stepsto_p \stat{\fram{\rho[x \mapsto \fAnd(\rho(y), \rho(z))]}{v} \cdot S}{\mu} \\
      \stat{\fram{\rho}{\inst[u]{\stmtAssign{x}{y \lor z}}} \cdot S}{\mu} \stepsto_p \stat{\fram{\rho[x \mapsto \fOr(\rho(y), \rho(z))]}{v} \cdot S}{\mu} \\
      \stat{\fram{\rho}{\inst[u]{\stmtAssign{x}{y.f}}} \cdot S}{\mu} \stepsto_p \stat{\fram{\rho[x \mapsto \mu(\rho(y))(f)]}{v} \cdot S}{\mu} \\
      \stat{\fram{\rho}{\inst[u]{\stmtAssign{x.f}{y}}} \cdot S}{\mu} \stepsto_p \stat{\fram{\rho}{v} \cdot S}{\mu[\rho(x) \mapsto [f \mapsto \rho(y)]]} \\
      \stat{\fram{\rho}{\inst[u]{\stmtMain}} \cdot S}{\mu} \stepsto_p \stat{\fram{\rho_0}{v} \cdot S}{\mu}
    \end{gather*}
  \end{footnotesize}
  \caption{Small-step semantics rules that hold when $\arcsto{p}{u}{v}$. \;
    $\dagger{}$ This particular rule only applies if either $a = \qm$ or
    $\rho_1(y) \in \fConc(a)$ (see Section~\ref{sec:static}).}
  \label{fig:semantics}
\end{figure}
The rules rely on the following helper functions:
\begin{align*}
  \fNew        & : \sMem_p \to \sVal \setminus \{0\} & \fNew(\mu)           & = 1 + \max(\{0\} \union \dom(\mu))                     \\
  \fBranch     & : \sVal \times \sVar \to \sInst                  & \fBranch(n, x)          & = \stmtIf{x}\text{ if }n > 0\text{; }\stmtElse{x}\text{ otherwise} \\
  \fAnd        & : \sVal \times \sVal \to \sVal      & \fAnd(n_1, n_2)      & = n_2\text{ if }n_1 > 0\text{; }n_1\text{ otherwise}   \\
  \fOr         & : \sVal \times \sVal \to \sVal      & \fOr(n_1, n_2)       & = n_1\text{ if }n_1 > 0\text{; }n_2\text{ otherwise}
\end{align*}

Notably, $\stmtBranch{y}$ steps to the $\stmtIf{y}$ node when $y$ is non-null and $\stmtElse{y}$ when $y$ is null. Additionally, if a procedure call's argument disagrees with its parameter annotation, then it will get stuck (rule~\ref{rule:call} for states); otherwise, the call statement will safely step to the procedure's body. In contrast, the semantics will get stuck if a return value does not agree with the procedure's return annotation.

\section{A Static Null-Pointer Analysis for PICL}\label{sec:static}
In this section, we formalize a static null-pointer analysis, called NPA, for PICL on which we will build GNPA. Here, we will only consider completely annotated programs, $\sAnn = \{ \aNa ,~ \aNN \}$. Therefore, we use a ``prime'' symbol for sets like $\sInst' \subseteq \sInst$ to indicate that this is not the whole story. We present NPA's semilattice of abstract values, flow function, fixpoint algorithm, and how the analysis uses the results from the fixpoint algorithm to report warnings to the user.

\subsection{Semilattice of Abstract Values} \label{sec:static-semilattice}
The set of abstract values $\sAbst = \{\aNa,~ \aNu,~ \aNN\}$ make up the finite semilattice defined in Figure~\ref{fig:semilattice}.
The partial order $\under~ \subseteq \sAbst \times \sAbst$ given is
\begin{mathpar}
\aNu \under \aNa
\and
\aNN \under \aNa
\and
\forall.~ l \in \sAbst ~.~ l \under l.
\end{mathpar}
The join function $\join : \sAbst \times \sAbst \to \sAbst$ induced by the partial order is:
\begin{mathpar}
\aNu ~\join~ \aNN = \aNa
\and
\forall.~ l \in \sAbst ~.~ l ~\join~ \aNa = \aNa
\and
\forall.~ l \in \sAbst ~.~ l ~\join~ l = l
\end{mathpar}
Clearly, $\aNa$ is the top element $\top$.
\begin{figure}
  \centering
  \begin{tikzpicture}
    \node (t) {\aNa};
    \node (a) [below left = 0.5cm and -0.5cm of t] {\aNu};
    \node (b) [below right = 0.5cm and -0.5cm of t] {\aNN};

    \draw (t) -- (a);
    \draw (t) -- (b);
  \end{tikzpicture}
  \caption{The \sAbst{} semilattice.}
  \label{fig:semilattice}
\end{figure}
Next, we relate this semilattice to $\sVal$ via a concretization function
$\fConc : \sAbst \to \powerset^{+}(\sVal)$:
  \[
    \fConc(\aNa) = \sVal \qc
    \fConc(\aNu) = \{0\} \qc
    \fConc(\aNN) = \sVal \setminus \{0\},
  \]
which satisfies the property $\forall.~ l_1, l_2 \in \sAbst ~.~ l_1 \under~ l_2 ~\iff~ \fConc(l_1) \subseteq \fConc(l_2)$.

\subsection{Flow Function}
Similar to how we use $\sEnv$ to represent mappings from variables to concrete values, we will use $\sigma \in \sMap = \sVar \pto \sAbst$ to represent mappings from variables to abstract values---\emph{abstract states}. Then, we extend the semilattice's partial order relation to abstract states $\sigma_1,
  \sigma_2 \in \sMap$:
\[
  \sigma_1 \under \sigma_2
  \quad\iff\quad
  \forall.~ x \in \sVar ~.~ \sigma_1(x) \under \sigma_2(x)
\]

\noindent We also extend the join operation to abstract states $\sigma_1, \sigma_2 \in \sMap$:
        \[
          (\sigma_1 \join \sigma_2)(x) =
          \begin{cases}
            a \join b        & \text{if $\sigma_1(x) = a$ and $\sigma_2(x) = b$}          \\
            a                & \text{if $\sigma_1(x) = a$ and $\sigma_2(x)$ is undefined} \\
            b                & \text{if $\sigma_1(x)$ is undefined and $\sigma_2(x) = b$} \\
            \text{undefined} & \text{otherwise}.
          \end{cases}
        \]

The NPA's flow function $\fFlow : \sInst' \times \sMap \to \sMap$ is defined in Figure~\ref{fig:flow}. Note, $\sigma_0 = \set{x \mapsto \aNu}{x \in \sVar}$. Also, we omit the $\stmtReturn{y\code{@}a}$ case because it does not have CFG successors in a well-formed program.

\begin{figure}[]
  \centering
  \begin{align*}
    \fFlow(\stmtAssign{x}{y}, \sigma)        & = \sigma[x \mapsto \sigma(y)]            \\
    \fFlow(\stmtBranch{x}, \sigma)           & = \sigma                                 \\
    \fFlow(\stmtIf{x}, \sigma)               & = \sigma[x \mapsto \aNN]                 \\
    \fFlow(\stmtElse{x}, \sigma)             & = \sigma[x \mapsto \aNu]                 \\
    \fFlow(\stmtCall{x}{m\code{@}a}{y\code{@}b}, \sigma) & = \sigma[x \mapsto a]                    \\
    \fFlow(\stmtProc{m\code{@}a}{y\code{@}b}, \sigma)    & = \sigma_0[y \mapsto b]                  \\
    \fFlow(\stmtAssign{x}{\vnull}, \sigma)             & = \sigma[x \mapsto \aNu]                 \\
    \fFlow(\stmtAssign{x}{\exprNew{\overline{f}}}, \sigma)       & = \sigma[x \mapsto \aNN]                 \\
    \fFlow(\stmtAssign{x}{y \land z}, \sigma) & = \begin{cases}
      \sigma[x \mapsto \aNu] & \text{if }\aNu \in \{\sigma(y), \sigma(z)\} \\
      \sigma[x \mapsto \aNa] & \text{if }\aNa \in \{\sigma(y), \sigma(z)\} \\
      \sigma[x \mapsto \aNN] & \text{otherwise}
    \end{cases} \\
    \fFlow(\stmtAssign{x}{y \lor z}, \sigma) & = \begin{cases}
      \sigma[x \mapsto \aNN] & \text{if }\aNN \in \{\sigma(y), \sigma(z)\} \\
      \sigma[x \mapsto \aNa] & \text{if }\aNa \in \{\sigma(y), \sigma(z)\} \\
      \sigma[x \mapsto \aNu] & \text{otherwise}
    \end{cases} \\
    \fFlow(\stmtAssign{x}{y.f}, \sigma)           & = \sigma[x \mapsto \aNa][y \mapsto \aNN] \\
    \fFlow(\stmtFieldAssign{x}{f}{y}, \sigma)        & = \sigma[x \mapsto \aNN]                 \\
    \fFlow(\stmtMain, \sigma)          & = \sigma_0
  \end{align*}
  \caption{All consequential cases of the flow function used by NPA.}
  \label{fig:flow}
\end{figure}

\subsubsection{Properties}
It can be shown that this flow function is monotonic: for any $\iota \in \sInst'$ and abstract states $\sigma_1, \sigma_2 \in \sMap$, if $\sigma_1 \under \sigma_2$ then $\fFlow\bb{\iota}(\sigma_1) \under \fFlow\bb{\iota}(\sigma_2)$. It can also be shown that the flow function is locally sound, \ie the flow function models the concrete semantics at each step. To express this property formally, we define the predicate $\pDesc(\rho, \sigma)$ on $\sEnv \times \sMap$, which says that the abstract state $\sigma$ ``describes'' the concrete environment $\rho$:
\[
  \pDesc(\rho, \sigma)
  \quad\iff\quad
  \text{for all $x \in \sVar ~.~ \rho(x) \in \fConc(\sigma(x))$.}
\]
Then, if
 $\stat{S' \cdot \fram{\rho}{\inst[v]{\iota}} \cdot S}{\mu}
  \stepsto_p
  \stat{\fram{\rho'}{v'} \cdot S}{\mu'}$,
it must be the case that
\[
  \pDesc(\rho, \sigma)
  \quad\implies\quad
  \pDesc(\rho', \fFlow\bb{\iota}(\sigma))
  \qq{for all}
  \sigma \in \sMap.
\]

\subsection{Fixpoint Algorithm}
This brings us to Algorithm~\ref{alg:kildall} \cite{kildall1973unified}, which is used to analyze a program and compute whether each program variable is $\aNa$, $\aNN$, or $\aNu$ at each program point (the program results $\pi$). More specifically, the algorithm applies the flow function to each program instruction recording or updating the results until a fixpoint is reached---\ie until the results stop changing (becoming more approximate). The algorithm will always reach a fixpoint (terminate), because $\fFlow$ is monotone and the height of the semilattice (Sec. \ref{sec:static-semilattice}) is finite. Note, the algorithm does not specify the order in which instructions are analyzed, because the order does not affect the results when $\fFlow$ is monotonic. An implementation may choose to analyze instructions in CFG order---following the directed edges of the CFG.

\begin{algorithm}
  \caption{Kildall's worklist algorithm}
  \label{alg:kildall}
  \begin{algorithmic}[1]
    \Function{Kildall}{$\fFlow, \join, p$}
      \State $\pi \gets \set{v \mapsto \varnothing}{v \in \sVert_p}$
      \State $V \gets \sVert_p$ \Comment{$V \subseteq \sVert_p$}
      \While{$V \neq \varnothing$} \label{line:while}
        \State $\inst[v]{\iota} \gets \text{an element of $V$}$ \Comment{$v \in V$ and $\iota = \fInst_p(v)$}
        \State $V \gets V \setminus \{v\}$ \Comment{$v \notin V$}
        \State $\sigma \gets \pi(v)$ \label{line:sigma}
        \State $\sigma' \gets \fFlow\bb{\iota}(\sigma)$ \label{line:flow}
        \For{$\arcsto{p}{v}{u}$} \label{line:for} \Comment{$u \in \sVert_p$}
          \If{$\sigma' \join \pi(u) \neq \pi(u)$} \label{line:if} \Comment{think of as $\sigma' \not\under \pi(u)$}
            \State{$\pi(u) \gets \pi(u) \join \sigma'$} \label{line:join}
            \State{$V \gets V \union \{u\}$} \label{line:enqueue}
          \EndIf \label{line:if-end}
        \EndFor \label{line:for-end}
      \EndWhile \label{line:while-end}
      \State \textbf{return} $\pi$
    \EndFunction
  \end{algorithmic}
\end{algorithm}

\subsection{Safety Function \& Static Warnings}
Next, we present a way to use analysis results $\pi$ produced by the fixpoint algorithm to determine whether to accept or reject a given program. Our goal is to ensure that when we run the program, it will not get stuck; that is, for any state $\xi$ that the program
reaches, we want to ensure that either $\xi$ is a final state $\stat{E \cdot \nil}{\mu}$ or there is another state $\xi'$ such that $\xi \stepsto_p \xi'$. To do this, we define the safety function $\fSafe\bb{\iota}(x) : \sInst' \times
\sVar \to \sAbst$, which returns the abstract value representing the set of ``safe'' values $x$ can take on before $\iota$ is executed. Figure~\ref{fig:safe} gives a few representative cases for $\fSafe$, and in all the cases not shown $\fSafe$ returns \aNa. In particular, a procedure call's argument must adhere to the procedure's parameter annotation, a return value must adhere to its corresponding return annotation, and all field accesses must have non-null receivers. Therefore, the safety function guards against all undefined behavior.

\begin{figure}[]
  \centering
  \begin{align*}
    \fSafe(\stmtCall{x}{m\code{@}a}{y\code{@}b}, y) & = b    \\
    \fSafe(\stmtReturn{y\code{@}a}, y)        & = a    \\
    \fSafe(\stmtAssign{x}{y.f}, y)           & = \aNN \\
    \fSafe(\stmtFieldAssign{x}{f}{y}, x)        & = \aNN
  \end{align*}
  \caption{All nontrivial cases of the safety function.}
  \label{fig:safe}
\end{figure}

\subsubsection{Static Warnings}

\noindent Now, we can state the meaning of a valid program $p \in \sProg'$:
\[
  \text{for all}\quad
  \inst[v]{\iota} \in \sVert_p
  \qand
  x \in \sVar
  ~.\quad
  \pi(v) = \sigma
  \quad\implies\quad
  \sigma(x) \under \fSafe\bb{\iota}(x)
\]
\[
  \text{where} \quad \pi = \fKildall(\fFlow, \join, p).
\]
That is, NPA emits static warnings when the fixpoint results disagree, according to the partial order $\under$, with the safety function. Also, we prove in Section \ref{sec:static-properties} that a valid program does not get stuck.

\subsection{Soundness of NPA} \label{sec:static-properties}
As discussed above, PICL's semantics are designed to get stuck when procedure annotations are violated or when null objects are dereferenced. Therefore, informally \emph{soundness} says that a valid program does not get stuck during execution. Formally, soundness is defined with progress and preservation statements. Before their statement we must first define the notion of valid states to complement our definition of valid programs:
\[
  \text{Let}~ p \in \sProg'. ~\text{A state}~ \xi =
  \stat{\fram{\rho_1}{v_1} \cdot \fram{\rho_2}{v_2} \cdots \fram{\rho_n}{v_n}
  \cdot \nil}{\mu} \in \sState_p ~\text{is \emph{valid} if}
\]
\[
  \text{for all}\quad
  1 \leq i \leq n
  ~.\quad
  \pDesc(\rho_i, \pi(v_i))
  \quad\text{where}~ \pi = \fKildall(\fFlow, \join, p).
\]
A state is \emph{valid} if it is described by the static analysis results $\pi$.

\begin{proposition}[static progress]\label{prop:static-progress}
  Let $p \in \sProg'$ be valid. If $\xi = \stat{E_1 \cdot E_2 \cdot S}{\mu} \in
    \sState_p$ is valid then $\xi \stepsto_p \xi'$ for some $\xi' \in \sState_p$.
\end{proposition}

\begin{proposition}[static preservation]\label{prop:static-preservation}
  Let $p \in \sProg'$ be valid. If $\xi \in \sState_p$ is valid and $\xi
    \stepsto_p \xi'$ then $\xi'$ is valid.
\end{proposition}

\section{Gradual Null-Pointer Analysis}\label{sec:gradual}
In this section, we derive GNPA from NPA, presented previously (Sec. \ref{sec:static}). We proceed following the Abstracting Gradual Typing methodology introduced by Garcia \etal~\cite{garcia2016abstracting} in the context of gradual type systems, adapting it to fit the concepts of static analysis.

We present the GNPA's lifted semilattice (Sec. \ref{sec:gradual-semilattice}), flow and safety functions (Sec. \ref{sec:gradual-flow-safety}), and fixpoint algorithm (Sec. \ref{sec:gradual-fixpoint}). We also discuss how static (Sec. \ref{sec:gradual-warnings}) and run-time warnings (Sec. \ref{sec:gradual-runtime}) are generated by the analysis. Finally, Section \ref{sec:gradual-properties} establishes the main properties of GNPA.

Note, here, annotations may be missing, so we extend our set of annotations with $\qm$: $\sAnn = \{\aNN,~\aNa\} \union \{\qm\}$.

\subsection{Lifting the Semilattice} \label{sec:gradual-semilattice}

In this section, we lift the semilattice ($\sAbst$, $\under$, $\join$) (Sec. \ref{sec:static-semilattice}) by following the Abstracting Gradual Typing (AGT) framework \cite{garcia2016abstracting}. First, we extend the set of semilattice elements $\sAbst$ to the new set $\Grad{\sAbst} \supseteq \sAbst$:
\[
 \Grad{\sAbst} = \sAbst \union \{\qm\} \union \set{\qm[a]}{a \in \sAbst} =
\]
$$
 \{ \aNa,~ \aNN,~ \aNu,~ \qm,~ \qm[\aNN],~ \qm[\aNu] \}.
$$

Note that we equate the elements $\qm[\aNa]$ and $\aNa$ in $\Grad{\sAbst}$. In Section \ref{sec:gradual-qmsemantics}, we give the semantics of the new lattice elements resulting in $\top = \qm[\aNa] = \aNa$. If $\sAbst$ had a bottom element $\bot$, then $\bot = \qm[\bot]$ similarly.

The join $\join$ and partial order $\under$ are also lifted to their respective counterparts $\Gradjoin$ (Sec. \ref{sec:gradual-join}) and $\Gradunder$ (Sec. \ref{sec:gradual-under}). The resulting lifted semilattice $(\Grad{\sAbst}, \Gradjoin)$ with lifted relation $\Gradunder$ underpins the optimism in GNPA.

\subsubsection{Giving Meaning to Missing Annotations} \label{sec:gradual-qmsemantics}
A straightforward way to handle $\qm$ would be to make it the top element $\qm = \top$ or the bottom element $\qm = \bot$ of NPA's semilattice. However, neither choice is sufficient for our goal:
\begin{itemize}
  \item If $\qm = \bot$, then $\qm \under a$ for all $a \in
    \sAbst$ and $\fConc(\bot) = \varnothing$. As a result, if the return annotation of a procedure was $\qm$, then we could use the return value in any context without the analysis giving a warning. But, anytime an initialized variable is checked against the $\qm$ annotation, such as checking the non-null return value $y$ against the $\qm$ return annotation $\aNN \under \qm$, the check will fail as $a \not \under \qm$ for all $a \in \sAbst ~.~ a \neq \bot$.
  \item If we let $\qm = \top$ then we have $a \under \qm$ for all $a \in
    \sAbst$. Therefore, we can pass any argument to a parameter annotated
    as $\qm$ without the static part of GNPA giving a warning. But, if the return
    annotation of that procedure is $\qm$, then the analysis will produce false positives in caller contexts wherever the return value is dereferenced. In other words, our analysis would operate exactly as \code{PolyNull} for the example in Fig. \ref{fig:java-safereverse}, which is not ideal.
\end{itemize}

Our goal is to construct an analysis system that does not produce false positive static warnings when a developer omits an annotation. To achieve this, we draw on work in gradual typing \cite{garcia2016abstracting}. We define the injective concretization function $\gamma
: \Grad{\sAbst} \to \powerset^{+}(\sAbst)$ where $\Grad{\sAbst} \supseteq \sAbst$ is the lifted semilattice element set (Sec. \ref{sec:gradual-semilattice}):
\[
  \gamma(a) = \{a\}
  \qq{for}
  a \in \sAbst,
  \quad
  \gamma(\qm) = \sAbst,
  \qand  \gamma(\qm[a]) = \set{b \in \sAbst}{a \under b}.
\]
An element in $\sAbst$ is mapped to itself as it can only represent itself. In contrast, $\qm$ may represent any element in $\sAbst$ at all times to support optimism in all possible contexts. Further, $\qm[a]$ means ``$a$ or anything more general than it,'' in contrast to a gradual formula $\phi \land \qm$ that means ``$\phi$ or anything more specific than it'' \cite{bader2018gradual}. As a result, $\qm[a]$ does not play the intuitive role of ``supplying missing information,'' as it would in gradual verification. Instead, $\qm[a]$ is simply an artifact of our construction, which is why the only element of $\sAnn \setminus \sAbst$ is $\qm$.

Then, if $\gamma(\Grad{a}) \subseteq \gamma(\Grad{b})$ for some $\Grad{a}, \Grad{b} \in \Grad{\sAbst}$, we write $\Grad{a} \preciser \Grad{b}$ and say that \emph{$\Grad{a}$ is more precise than $\Grad{b}$}. Further, $\iota_1 \preciser \iota_2$ means that 1) the two instructions are equal except for their annotations, and 2) the annotations in $\iota_1$ are more precise than the corresponding annotations in $\iota_2$.

\subsubsection{Lifted Join $\Gradjoin$} \label{sec:gradual-join}

We begin by introducing a semilattice definition \cite{davey2002introduction}, which states that a semilattice is an algebraic structure $(S, \join)$ where for all $x, y, z \in S$ the following hold:
\begin{itemize}
  \item $x \join (y \join z) = (x \join y) \join z$ (associativity)
  \item $x \join y = y \join x$ (commutativity)
  \item $x \join x = x$ (idempotency)
\end{itemize}
Then, we write $x \under y$ when $x \join y = y$ and it can be shown this $\under$ is a partial order. Recall that NPA uses $\join$ in Algorithm~\ref{alg:kildall} to compute a fixpoint that describes the behavior of a program $p$. The fixpoint can only be reached when $\join$ is idempotent. Similarly, $\join$ must be commutative and associative so that program instructions can be analyzed in any order. Thus, our extended join operation $\Gradjoin : \Grad{\sAbst} \times \Grad{\sAbst} \to
\Grad{\sAbst}$ must be associative, commutative, and idempotent making $(\Grad{\sAbst}, \Gradjoin)$ a join-semilattice.

To define such a function we turn to insights from gradual typing \cite{garcia2016abstracting}. We define an abstraction function $\alpha : \powerset^{+}(\sAbst) \to
\Grad{\sAbst}$, which forms a Galois connection with $\gamma$:
\[
  \alpha(\Coll{a})
  = \gamma^{-1}\qty(
        \Intersect_{\substack{
            \Grad{b} \in \Grad{\sAbst} \\
            \gamma(\Grad{b}) \supseteq \Coll{a}
          }}
        \gamma(\Grad{b})
      )
\]
where, for $a \in \sAbst$, $\gamma^{-1}$ is:
\begin{mathpar}
\gamma^{-1}(\{a\}) = a
\and
\gamma^{-1}(\sAbst) = \qm
\and
\gamma^{-1}(\{b \in \sAbst : a \under b \}) = \qm[a].
\end{mathpar}
Then we define the join of $\Grad{a}, \Grad{b} \in \Grad{\sAbst}$ as follows:
\[
  \Grad{a} \Gradjoin \Grad{b}
  = \alpha(\set{a \join b}{\text{
      $a \in \gamma(\Grad{a})$
      and $b \in \gamma(\Grad{b})$}
  })
\]
For example,
\begin{flalign}
\aNN \Gradjoin \qm &= \alpha(\{ a \join b : a \in \{\aNN\} ~\text{and}~ b \in \sAbst \}) \label{join-eq-alpha-start}\\
&= \alpha(\{ \aNN,~ \aNa \}) \label{join-eq-alpha-set}\\
&= \gamma^{-1}\left( \gamma(\qm[\aNN]) \intersect \gamma(\qm) \right) \label{join-eq-intersect-gamma}\\
&= \gamma^{-1}\left( \{ \aNN,~ \aNa \} \intersect \sAbst \right) \label{join-eq-intersect-sets}\\
&= \gamma^{-1}\left( \{ \aNN,~ \aNa \} \right) \label{join-eq-intersect-result}\\
&= \qm[\aNN] \label{join-eq-result}
\end{flalign}
That is, the join of all the $\sAbst$ elements represented by $\aNN$ and $\qm$ results in the set $\{\aNN,~ \aNa\}$ (\ref{join-eq-alpha-start}, \ref{join-eq-alpha-set}). Applying $\alpha$ to this set is equivalent to applying $\gamma^{-1}$ to $\gamma(\qm[\aNN]) \intersect \gamma(\qm)$ (\ref{join-eq-intersect-gamma}); because, the only $\Grad{\sAbst}$ elements that represent both $\aNN$ and $\aNa$ are $\qm[\aNN]$ and $\qm$. The intersection of $\gamma(\qm[\aNN])$ and $\gamma(\qm)$ is $\{\aNN,$ $\aNa\}$ (\ref{join-eq-intersect-sets}, \ref{join-eq-intersect-result}), so we are really applying  $\gamma^{-1}$ to $\{\aNN,~ \aNa\}$ (\ref{join-eq-intersect-result}). Therefore, $\aNN \Gradjoin \qm = \qm[\aNN]$ (\ref{join-eq-result}).
Notice, the intersection of the representative sets $\gamma(\qm[\aNN])$ and $\gamma(\qm)$ of $\{\aNN,~\aNa\} = \Coll{a}$ is used to find the most precise element in $\Grad{\sAbst}$ that can represent $\Coll{a}$.

Now we return to the properties of $\Gradjoin$. Since $\join$ is commutative, we have that $\Gradjoin$ is commutative. Idempotency is also not too onerous: it is equivalent to the condition that every element of $\Grad{\sAbst}$ represents a subsemilattice of $\sAbst$. That is, for every $\Grad{a} \in \Grad{\sAbst}$ and $a_1, a_2 \in \gamma(\Grad{a})$, we must have $a_1 \join a_2 \in \gamma(\Grad{a})$. This is true by construction. Associativity is tricky and motivates our complex definition of $\Grad{\sAbst}$. Ideally, $\Grad{\sAbst}$ would be defined simply as $\sAbst \union \{\qm\}$, however in this case $\Gradjoin$ is not associative:
  \begin{align*}
    \aNu \Gradjoin (\aNN \Gradjoin \qm)
     & = \aNu \Gradjoin \qm \\
     & = \qm                                  \\
     & \neq \aNa                              \\
     & = \aNa \Gradjoin \qm \\
     & = (\aNu \Gradjoin \aNN) \Gradjoin \qm.
  \end{align*}
Fortunately, our definition of $\Grad{\sAbst}$ which also includes the intermediate optimistic elements $\qm[\aNN]$ and $\qm[\aNu]$ results in an associative $\Gradjoin$ function and a finite-height semilattice $(\Grad{\sAbst},~ \Gradjoin)$. Figure~\ref{fig:semilattice-lifted} shows the semilattice structure induced by $\Gradjoin$.

\begin{figure}
  \centering
  \begin{tikzpicture}
    \node (t) {$\aNa$};
    \node (aq) [below left = 0.5cm and -0.7cm of t] {$\qm[\aNu]$};
    \node (bq) [below right = 0.5cm and -0.7cm of t] {$\qm[\aNN]$};
    \node (q) [below right = 0.5cm and 0cm of aq] {$\qm$};
    \node (a) [below left = 0.5cm and -0.7cm of aq] {$\aNu$};
    \node (b) [below right = 0.5cm and -1cm of bq] {$\aNN$};

    \draw (t) -- (aq);
    \draw (t) -- (bq);
    \draw (aq) -- (a);
    \draw (aq) -- (q);
    \draw (bq) -- (q);
    \draw (bq) -- (b);
  \end{tikzpicture}
  \caption{The semilattice structure induced by the lifted join $\Gradjoin$. Specifically, this is the Hasse diagram of the partial order
    $\set{(\Grad{a}, \Grad{b})}{\Grad{a} \Gradjoin \Grad{b} = \Grad{b}}$.}
  \label{fig:semilattice-lifted}
\end{figure}

\subsubsection{Lifted Order $\Gradunder$}\label{sec:gradual-under}

Now it is fairly straightforward to construct $\Gradunder$. Recall, NPA emits static warnings when the fixpoint results disagree with the safety function, according to the partial order $\under$. The fixpoint results and the safety function now return elements in $\Grad{\sAbst}$, so we lift $\under$ to $\Gradunder ~\subseteq \Grad{\sAbst} \times \Grad{\sAbst}$ using the concretization function $\gamma$:
\[
  \Grad{a} \Gradunder \Grad{b}
  \quad\iff\quad \exists ~.\quad
  a \in \gamma(\Grad{a})
  \qand
  b \in \gamma(\Grad{b})
  \qq{such that}
  a \under b
  \qq{for} \Grad{a}, \Grad{b} \in \Grad{\sAbst}.
\]
Figure~\ref{fig:order} gives the lifted order relation $\Gradunder$ in graphical form.

\begin{figure}[]
  \centering
  \begin{tikzpicture}[>=stealth]
    \graph[counterclockwise=6]{
      t/"$\top$";
      aq/"$\qm[A]$";
      a/"$A$";
      q/"$\qm$";
      b/"$B$";
      bq/"$\qm[B]$";

      aq -> t;
      a -> t;
      q -> t;
      b -> t;
      bq -> t;

      t -> aq;
      a -> aq;
      q -> aq;
      bq -> aq;

      aq -> a;
      q -> a;

      t -> q;
      aq -> q;
      a -> q;
      b -> q;
      bq -> q;

      q -> b;
      bq -> b;

      t -> bq;
      b -> bq;
      q -> bq;
      aq -> bq;

      a -> bq;
      b -> aq;
    };
  \end{tikzpicture}
  \caption{The lifted partial order, where each
    directed edge $\Grad{a} \to \Grad{b}$ means $\Grad{a} \Gradunder \Grad{b}$.
    (Self-loops are omitted). Here, \aNa{} is abbreviated $\top$, and \aNu{} and
    \aNN{} are abbreviated $A$ and $B$ respectively.}
  \label{fig:order}
\end{figure}

The $\Grad{\under}$ predicate is a maximally permissive version of the $\under$ predicate for $\qm[\aNN]$, $\qm[\aNu]$, and $\qm$. For example, $\qm \Gradunder \aNN$ since $\gamma(\qm) = \{\aNN,~ \aNu,~ \aNa\}$, $\gamma(\aNN) = \{\aNN\}$, and $\aNN \under \aNN$. By similar reasoning, $\aNN \Gradunder \qm$. In fact, $\qm \Gradunder a \Gradunder \qm$, $\qm[\aNN] \Gradunder a \Gradunder \qm[\aNN]$, and $\qm[\aNu] \Gradunder a \Gradunder \qm[\aNu]$ for $a \in \sAbst$. So, clearly $\Gradunder$ is not a partial order. The $\Gradunder$ predicate must be maximally permissive to support the optimism used in the \code{safeReverse} example from Figure \ref{fig:java-safereverse} (Sec. \ref{sec:action-falsepos}): calls to \code{safeReverse} with null and non-null arguments are valid and dereferences of its return values are also valid.
However, $\Gradunder$ is the same as $\under$ when both of its arguments come from $\sAbst$, \eg $\aNN \Gradunder \aNa$ and $\aNa \not \Gradunder \aNN$. This allows our gradual analysis to apply NPA where annotations are complete enough to support it.

\subsubsection{Properties}\label{subsec:semilattice-properties}
We previously mentioned some of the properties which ($\Grad{\sAbst}$, $\Gradjoin$) satisfy. Here, we formally state them, and their proofs can be found in the appendix of the full version of this paper.

\begin{proposition}\label{prop:semilattice}
  $(\Grad{\sAbst},~ \Gradjoin)$ is a semilattice; in other words, $\Gradjoin$ is associative, idempotent, and commutative.
\end{proposition}

\begin{proposition}\label{prop:finite-height}
  If the height of $(\sAbst,~ \join)$ is $n > 0$, then the height of
  $(\Grad{\sAbst},~ \Gradjoin)$ is $n + 1$ (in particular, $(\Grad{\sAbst},~ \Gradjoin)$ has finite height).
\end{proposition}

\subsection{Lifting the Flow \& Safety Functions} \label{sec:gradual-flow-safety}
Now both instructions and abstract states ($\Grad{\sigma} \in \Grad{\sMap} = \sVar \pto \Grad{\sAbst}$) may contain optimistic abstract values. Therefore, similar to lifting the join $\Gradjoin$, we follow the AGT \emph{consistent function lifting} approach \cite{garcia2016abstracting} when defining GNPA's
flow function $\Grad{\fFlow} : \sInst \times \Grad{\sMap} \to \Grad{\sMap}$ for this new domain.

Specifically, for $\iota \in \sInst$ and
$\Grad{\sigma} = \set{x \mapsto \Grad{a}_x}{x \in \sVar} \in \Grad{\sMap}$, we
define
\begin{flalign*}
   & \Grad{\fFlow}\bb{\stmtCall{z}{m\code{@}a}{y\code{@}b}}(\Grad{\sigma}) = \{x \mapsto \alpha(\{(\fFlow\bb{\stmtCall{z}{m\code{@}a'}{y\code{@}b'}}(\sigma'))(x)
   \\
   & \hspace{135pt} : a' \in \gamma(a) \wedge b' \in \gamma(b) \wedge \sigma' \in \Sigma\}) :x \in \sVar\} \\
   & \Grad{\fFlow}\bb{\stmtProc{m\code{@}a}{y\code{@}b}}(\Grad{\sigma}) = \{x \mapsto \alpha(\{(\fFlow\bb{\stmtProc{m\code{@}a'}{y\code{@}b'}}(\sigma'))(x)
   \\
   & \hspace{135pt} : a' \in \gamma(a) \wedge b' \in \gamma(b) \wedge \sigma' \in \Sigma\}) :x \in \sVar\} \\
   & \Grad{\fFlow}\bb{\iota}(\Grad{\sigma}) = \set{x \mapsto \alpha(\set{(\fFlow\bb{\iota}(\sigma'))(x)}{\sigma' \in \Sigma})}{x \in \sVar}
   \quad \text{otherwise}\\
   & \\
   & \text{where} \quad \Sigma = \set{\set{x \mapsto a_x}{x \in \sVar}}{a_x \in \gamma(\Grad{a}_x)\text{ for all }x \in \sVar}.
\end{flalign*}
Note that the procedure call and procedure entry instructions are the only instructions in $\fFlow$'s domain that may contain $\qm$ annotations, so the corresponding $\fFlow$ rules are lifted with respect to those annotations. Similarly, all rules are lifted with respect to their abstract states.

Recall that we defined the predicate $\pDesc$ on $\sEnv \times \sMap$ to express the local soundness of $\fFlow$. For $\Grad{\fFlow}$, we lift $\pDesc$ to $\Grad{\pDesc}$ on $\sEnv \times \Grad{\sMap}$ such that it is maximally permissive like the $\Gradunder$ predicate:
\[
  \Grad{\pDesc}(\rho, \Grad{\sigma})
  \quad\iff\quad
  \pDesc(\rho, \sigma)\text{ for some }\sigma \in \Sigma
\]
\[
\text{where}~ \Sigma ~\text{is constructed in the same way as for}~ \Grad{\fFlow}.
\]

\noindent Finally, we again follow the consistent function lifting methodology to construct $\Grad{\fSafe} : \sInst \times \sVar \to \Grad{\sAbst}$ from $\fSafe : \sInst' \times \sVar \to \sAbst$:
\begin{flalign*}
    & \Grad{\fSafe}\bb{\stmtCall{z}{m\code{@}a}{y\code{@}b}}(x) = \alpha(\set{\fSafe\bb{\stmtCall{z}{m\code{@}a'}{y\code{@}b'}}(x)}{a' \in \gamma(a) \wedge b' \in \gamma(b)})\\
    & \Grad{\fSafe}\bb{\stmtProc{m\code{@}a}{y\code{@}b}}(x) = \alpha(\set{\fSafe\bb{\stmtProc{m\code{@}a'}{y\code{@}b'}}(x)}{a' \in \gamma(a) \wedge b' \in \gamma(b)})\\
    & \Grad{\fSafe}\bb{\stmtReturn{y\code{@}a}}(x) = \alpha(\set{\fSafe\bb{\stmtReturn{y\code{@}a'}}(x)}{a' \in \gamma(a)})\\
    & \Grad{\fSafe}\bb{\iota}(x) = \alpha(\fSafe\bb{\iota}(x)) \quad \text{otherwise}
\end{flalign*}

Other than the casewise-defined $\fFlow$ rules for $\land$ and $\lor$, the lifted $\Grad{\fFlow}$ and $\Grad{\fSafe}$ functions simplify down to the same computation rules as $\fFlow$ and $\fSafe$ as shown in Figure~\ref{fig:flow} and Figure~\ref{fig:safe} respectively, replacing $\fFlow$ with $\Grad{\fFlow}$ and $\fSafe$ with $\Grad{\fSafe}$.

\subsection{Lifting the Fixpoint Algorithm} \label{sec:gradual-fixpoint}
To lift the fixpoint algorithm, we simply plug $\Grad{\fFlow}$ and
$\Gradjoin$ into Algorithm~\ref{alg:kildall} to compute $\Grad{\pi} =
\fKildall(\Grad{\fFlow}, \Gradjoin, p) : \sVert_p \to \Grad{\sMap}$ for any $p \in \sProg$.

\subsection{Static Warnings} \label{sec:gradual-warnings}
Using the lifted safety function, we say that a partially-annotated program $p
\in \sProg$ is \emph{statically valid} if
\[
  \text{for all}\quad
  \inst[v]{\iota} \in \sVert_p
  \qand
  x \in \sVar
  \qc
  \Grad{\pi}(v) = \Grad{\sigma}
  \quad\implies\quad
  \Grad{\sigma}(x) \Gradunder \Grad{\fSafe}\bb{\iota}(x)
\]
\[
  \text{where}~ \Grad{\pi} = \fKildall(\Grad{\fFlow}, \Gradjoin, p).
\]
Each piece of GNPA's static system ($(\Grad{\sAbst}, \Gradjoin)$, $\Gradunder$, $\Grad{\fFlow}$, $\Grad{\fSafe}$, and the fixpoint algorithm) is designed to be maximally optimistic for missing annotations. Therefore, the resulting system will not produce false positive warnings due to missing annotations. The system is also designed to apply NPA where annotations are available to support it, so it will still warn about violations of procedure annotations or null object dereferences where possible. See Section \ref{sec:action-falsepos} for more information.

\subsection{Dynamic Checking} \label{sec:gradual-runtime}
GNPA's static system reduces false positive warnings at the cost of soundness. For example,  as in Section \ref{sec:action-falseneg}, the analysis may assume a variable with a $\qm$ annotation is non-null to satisfy an object dereference when the variable is actually null. In order to avoid false negatives and ensure that our gradual analysis is sound, we modify the semantics of PICL to insert run-time checks where the analysis may be unsound.
That is, if $p$ is \emph{statically valid} and there are program points $\inst[v]{\iota}$ such that
\[
  a \not\under \bigjoin \gamma(\Grad{\fSafe}\bb{\iota}(x))
  \qq{for some}
  x \in \sVar
  \qand
  a \in \gamma((\Grad{\pi}(v))(x)),
\]
then a run-time check must be inserted at those points to ensure the value of $x$ is in $\fConc(\bigjoin \gamma(\Grad{\fSafe}\bb{\iota}(x)))$.

More precisely, we define a dedicated error state $\cError$ and expand the set
of run-time states to be $\Grad{\sState}_p = \sState_p \union \{\cError\}$. Then
we define a restricted semantics $\Gradstepsto_p$ on $\Grad{\sState}_p \times
\Grad{\sState}_p$ as follows. Let $\xi \in \sState_p$. If
\[
  \xi = \stat{\fram{\rho}{\inst{\iota}} \cdot S}{\mu}
  \qand
  \lnot\Grad{\pDesc}(\rho, \set{x \mapsto \Grad{\fSafe}\bb{\iota}(x)}{x \in \sVar})
\]
then $\xi \Gradstepsto_p \cError$. If there is some $\xi' \in \sState_p$ such that $\xi \stepsto_p \xi'$, then $\xi \Gradstepsto_p \xi'$. Otherwise, there is no $\Grad{\xi}' \in \Grad{\sState}_p$ such that $\xi \Gradstepsto_p \xi'$.

\subsection{Gradual Properties} \label{sec:gradual-properties}
GNPA is sound, \emph{conservative extension} of NPA---the static system is applied in full to programs with complete annotations, and adheres to the gradual guarantees inspired by Siek \etal~\cite{siek2015refined}. The gradual guarantees ensure losing precision is harmless, \ie increasing the number of missing annotations in a program does not break its validity or reducibility.

To formally present each property, we first extend the notion of a valid state. Let $p \in \sProg$. A state $\xi = \stat{\fram{\rho_1}{v_1} \cdot \fram{\rho_2}{v_2} \cdots \fram{\rho_n}{v_n} \cdot \nil}{\mu} \in \sState_p$ is valid if
\[
  \text{for all}\quad
  1 \leq i \leq n
  \qc
  \Grad{\pDesc}(\rho_i, \Grad{\pi}(v_i))
  \quad \text{where} \quad \Grad{\pi} = \fKildall(\Grad{\fFlow}, \Gradjoin, p).
\]
Then, for fully-annotated programs, GNPA and the modified semantics are conservative extensions of NPA and PICL's semantics, respectively.

\begin{proposition}[conservative static extension]\label{prop:conservative-static-extension}~\\
  If $p \in \sProg'$ then $\fKildall(\fFlow, \join, p) =
    \fKildall(\Grad{\fFlow}, \Gradjoin, p)$.
\end{proposition}

\begin{proposition}[conservative dynamic extension]\label{prop:conservative-dynamic-extension}
  Let $p \in \sProg'$ be valid, and let $\xi_1, \xi_2 \in \sState_p$. If $\xi_1$
  is valid then $\xi_1 \stepsto_p \xi_2$ if and only if $\xi_1 \Gradstepsto_p
    \xi_2$.
\end{proposition}

GNPA is sound, \ie valid programs will not get stuck during execution. However, programs may step to a dedicated \cError{} state when run-time checks fail. Soundness is stated with a progress and preservation argument.

\begin{proposition}[gradual progress]\label{prop:gradual-progress}
  Let $p \in \sProg$ be valid. If $\xi = \stat{E_1 \cdot E_2 \cdot S}{\mu} \in
  \sState_p$ is valid then $\xi \Gradstepsto_p \Grad{\xi}'$ for some
  $\Grad{\xi}' \in \Grad{\sState}_p$.
\end{proposition}

\begin{proposition}[gradual preservation]\label{prop:gradual-preservation}
  Let $p \in \sProg$ be valid. If $\xi \in \sState_p$ is valid and $\xi
  \Gradstepsto_p \xi'$ for some $\xi' \in \sState_p$, then $\xi'$ is valid.
\end{proposition}

Finally, GNPA satisfies both the static and dynamic gradual guarantees. Both of the guarantees rely on a definition of \emph{program precision}. Specifically, if programs $p_1$ and $p_2$ are identical except perhaps that some annotations in $p_2$ are $\qm$ where they are not $\qm$ in $p_1$, then we say that \emph{$p_1$ is more precise than $p_2$}, and write $p_1 \preciser p_2$.

Then, the \emph{static gradual guarantee} states that increasing the number of missing annotations in a valid program does not introduce static warnings (\ie break program validity).

\begin{proposition}[static gradual guarantee]\label{prop:static-gradual-guarantee}
  Let $p_1, p_2 \in \sProg$ such that $p_1 \preciser p_2$. If $p_1$ is
  statically valid then $p_2$ is statically valid.
\end{proposition}

The \emph{dynamic gradual guarantee} ensures that increasing the number of missing annotations in a program does not change the observable behavior of the program (\ie break program reducibility for valid programs).

\begin{proposition}[dynamic gradual guarantee]\label{prop:dynamic-gradual-guarantee}
  Let $p_1, p_2 \in \sProg$ be statically valid, where $p_1 \preciser p_2$. Let
  $\xi_1, \xi_2 \in \sState_{p_2}$. If $\xi_1 \Gradstepsto_{p_1} \xi_2$ then
  $\xi_1 \Gradstepsto_{p_2} \xi_2$.
\end{proposition}

Note, the small-step semantics $\Gradstepsto$ are designed to make the proofs of the aforementioned properties easier at the cost of easily implementable run-time checks. Therefore, we give the following proposition that connects a more implementable design to $\Gradstepsto$. That is, we can use the contrapositive of this proposition to implement more optimal run-time checks. Specifically, the na\"ive implementation would check each variable at each program point to make sure it satisfies the safety function for the instruction about to be executed. But Proposition~\ref{prop:static-progress} tells us that we only need to check variables at runtime when our analysis results don't already guarantee (statically) that they will satisfy the safety function.

\begin{proposition}[run-time checks]\label{prop:runtime-checks}
  Let $p \in \sProg$ be valid according to $\Grad{\pi} =
  \fKildall(\Grad{\fFlow}, \Gradjoin, p)$, and let $\xi =
  \stat{\fram{\rho}{\inst[v]{\iota}} \cdot S}{\mu} \in \sState_p$ be valid. If
  $\xi \Gradstepsto_p \cError$ then there is some $x \in \sVar$ and $a \in
  \gamma((\Grad{\pi}(v))(x))$ such that $a \not\under \bigjoin
  \gamma(\Grad{\fSafe}\bb{\iota}(x))$.
\end{proposition}

\section{Preliminary Empirical Evaluation}\label{sec:empirical}
In this section, we discuss the implementation of GNPA and two studies designed to evaluate its usefulness in practice. Preliminary evidence suggests that our analysis can be used at scale, produces fewer false positives than state-of-the-art tools, and eliminates on average more than half of the null-pointer checks Java automatically inserts at run time. These results illustrate an important practical difference between GNPA and other null-pointer analyses. While a sound static analysis can be used to prove the redundancy of run-time checks, and an unsound static analysis can be used to reduce the number of false positives, neither of those can do both at the same time. On the other hand, GNPA can both prove the redundancy of run-time checks and reduce reported false positives

\subsection{Research Questions}

We seek answers to the following questions:
\begin{enumerate}
  \item \label{q:prototype} Can a gradual null-pointer analysis be effectively implemented and used at scale?
  \item \label{q:false-positives} Does such a null-pointer analysis produce fewer false positives than industry-grade analyses?
  \item \label{q:eliminate-checks} Does the gradual null-pointer analysis perform significantly fewer null-pointer checks than the na\"ive approach of checking every dereference?
\end{enumerate}

\subsection{Prototype}

Facebook Infer provides a framework to construct static analyses that use abstract interpretation. We built a prototype of GNPA, called \emph{Graduator}, in this framework. Our prototype uses Infer's HIL intermediate language representation (IR). As a result, Graduator can be used to analyze code written in C, C++, Objective-C, and Java.

\begin{figure}
  \centering
  \begin{tikzpicture}
    \node (t) {$\aNa$};
    \node (n) [below = of t] {$\aNN$};

    \draw (t) -- (n);
  \end{tikzpicture}
  \qquad\qquad
  \begin{tikzpicture}[->,>=stealth']
    \node (t) {$\aNa$};
    \node (q) [below right=0.5cm and 0.5cm of t] {$\qm$};
    \node (n) [below=1.35cm of t] {$\aNN$};

    \draw (n) -> (q);
    \draw (n) -> (t);
    \draw (t) -> (q);
    \draw (q) -> (n);
    \draw (q) -> (t);
  \end{tikzpicture}
  \qquad\qquad
  \begin{tikzpicture}
    \graph[grid placement, wrap after=1]{
      t/"$\aNa$" -- q/"$\qm$" -- n/"$\aNN$";
    };
  \end{tikzpicture}
  \caption{\textit{Left:} The starting null-pointer semilattice for Graduator.
    \textit{Middle:} The lifted partial ordering, where each directed edge $\Grad{a} \to \Grad{b}$ means $\Grad{a} \Gradunder \Grad{b}$. (Self-loops are omitted.) \textit{Right:} The semilattice structure induced by the lifted join $\Gradjoin$.}
  \label{fig:smaller-lattice}
\end{figure}

The preceding case study (Secs. \ref{sec:language}--\ref{sec:gradual}) uses a base semilattice with three elements, \aNu, \aNN, and \aNa, in order to demonstrate that a semilattice lifting may contain additional intermediate optimistic elements, \qm[\aNu] and \qm[\aNN]. For simplicity, we implemented the semilattice from Figure \ref{fig:smaller-lattice}, along with its lifted variant, order relation and join function, in our prototype. This semilattice is the same as the base one in the case study except it does not contain \aNu: the initial static semilattice has only \aNN{} and \aNa{}, and the gradual semilattice only adds one additional \qm{} element. There are a couple other differences between our formalism and our Graduator prototype, one of which is that Graduator allows field annotations while our formalism does not.

Infer does not support modifying Java source code, so Graduator simply reports the locations where it should insert run-time checks rather than inserting them directly. In fact, Graduator may output any of the following:
\begin{itemize}
  \item \code{GRADUAL\_STATIC}---a static warning.
  \item \code{GRADUAL\_CHECK}---a location to check a possibly-null
    dereference.
  \item \code{GRADUAL\_BOUNDARY}---another location to insert a check,
    such as passing an argument to a method, returning from a method, or assigning a value to a field.
\end{itemize}
Since Java checks for null-pointer dereferences automatically, soundness is preserved. A more complete implementation of GNPA would insert run-time checks as part of the build process. As a result, some bugs may be caught earlier when the gradual analysis inserts checks at method boundaries and field assignments.

By implementing Graduator with Infer's framework, Graduator is guaranteed to operate at scale. We also evaluate Graduator on a number of open source repositories as discussed in Sections \ref{sec:empirical-static} and \ref{sec:empirical-runtime}.
Thus, the answer to RQ\ref{q:prototype} is yes.

\subsection{Static Warnings} \label{sec:empirical-static}

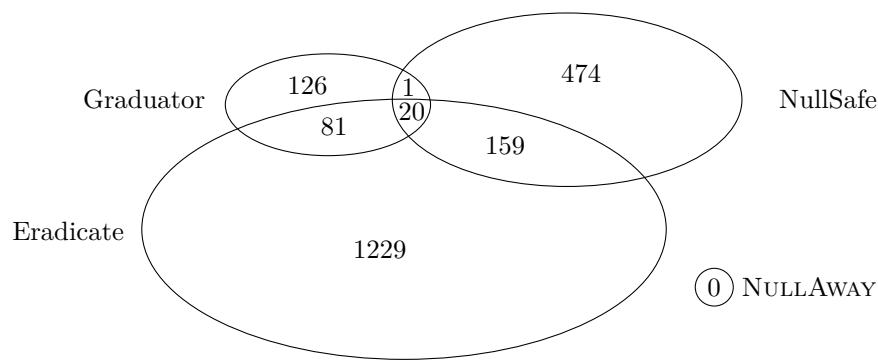
\begin{figure}[]
  \centering
  \begin{tikzpicture}[scale=0.05]
    \draw (28.3598,-34.6431) circle (69 and 34.5); 
    \draw (8.33837,-1.62742) circle (27 and 13.5); 
    \draw (71.2821,-0.296814) circle (46 and 23.0); 
    \draw (110,-50) circle (5 and 5); 

    \node at (-60,-35.0) {Eradicate};
    \node at (-40,0.0) {Graduator};
    \node at (140,0.0) {NullSafe};
    \node at (135,-50) {\textsc{NullAway}};

    \node at (110,-50) {0};

    \node at (22,-40.0) {1229};
    \node at (3,3.5) {126};
    \node at (75,6.5) {474};

    \node at (29.5,3) {1};
    \node at (55,-12.5) {159};
    \node at (10,-7.5) {81};

    \node at (30.4,-3.5) {20};
  \end{tikzpicture}
  \caption{The total number of static warnings reported by the three Infer null
    checkers, for all 15 repositories.}
  \label{fig:venn}
\end{figure}

To evaluate Graduator, we ran it on 15 of the 18 open-source Java repositories used to evaluate \textsc{NullAway} \cite{banerjee2019nullaway} (we excluded 3 of the 18 repositories because we were unable to successfully run Infer on them). We also ran \textsc{NullAway}, and Infer's existing null-pointer checkers Eradicate and NullSafe, on the repositories. Figure~\ref{fig:venn} shows the number of \emph{static} warnings produced by each of these three checkers: 1489 for Eradicate, 654 for NullSafe, 228 for Graduator, and 0 for \textsc{NullAway}, for a total of 2371.

Based on the \textsc{NullAway} paper (in which Uber states that in practice they have found no instances of null-pointer dereferences caused by their tool's unsoundness), it seems reasonable to assume that these repositories do not have null-pointer bugs, since \textsc{NullAway} itself reports no static warnings for these repositories. After examining all 2371 warnings ourselves, we found that all but 57 (50 from Eradicate only, 2 from Graduator only, and 5 from Eradicate and Graduator but not NullSafe) were false positives due to systematic imprecision in the analysis tools. We were unable to determine whether the remaining 57 warnings represent actual bugs or not.

Under this assumption, Graduator reports significantly fewer false positives than Infer's existing null-pointer checkers (although in this respect, it is of course outperformed by \textsc{NullAway}) (RQ\ref{q:false-positives}).
An interesting aspect of Figure~\ref{fig:venn} is how many warnings are produced by only one of the checkers: 1229 for Eradicate, 474 for NullSafe, and 126 for Graduator. Many of these warnings arose from generated and test case code.

\subsubsection{Generated Code}

Several of the 15 repositories generate code as part of their build process, and in some cases, the analysis tools gave warnings about the generated code. This accounts for
\begin{itemize}
  \item 380 of the warnings given by NullSafe alone,
  \item 356 of the warnings given by Eradicate alone,
  \item 130 of the warnings given by both Eradicate and NullSafe but not
    Graduator, and
  \item 8 of the warnings given by Graduator alone.
\end{itemize}
Graduator reports significantly fewer static warnings for generated code, because such code is typically unannotated and Graduator is designed to be optimistic when annotations are missing.

\subsubsection{Test Code}

It is reasonable to assume that test code does not contain null dereference bugs, because if it did, then those bugs would show up when the tests are run. Static warnings about test code account for
\begin{itemize}
  \item 384 of the warnings given by Eradicate alone, and
  \item 73 of the warnings given by both Eradicate and Graduator, but not NullSafe.
\end{itemize}
That is, Graduator reports fewer warnings for test code than Eradicate, but more than NullSafe. The NullSafe checker does not appear to treat test code specially, so it is unclear why NullSafe is performing better than Graduator for such code.

\subsubsection{Remaining False Positives}

The reader may wonder why Graduator reports any false positives on this codebase, since it intuitively seems that the static portion of a gradual analysis ought to be optimistic.  Examining the warnings given by Graduator, we see that none of the warnings are due to treating missing annotations pessimistically; instead, they are due to places where the analysis has whatever annotations it needs, but the analysis is imprecise in other respects.  For example, one common source of false positives is when a field is checked for null, then is read again.  Our original static analysis is limited in that it does not treat fields flow-sensitively, causing false positives that are independent of the choice to be gradual or not with respect to annotations.

\textsc{NullAway} avoids giving false positives on this same codebase, due to a combination of some unsound assumptions and a more precise analysis approach.  While our approach for deriving gradual program analyses focuses on retaining soundness through a combination of static and dynamic checks, incorporating more precise analysis techniques (e.g. a flow-sensitive treatment of fields, perhaps in combination with a gradual alias analysis) could eliminate more of these false positives.  In the meantime, our comparison to Eradicate and NullSafe is appropriate as these are the static analysis tools taking the most similar approach.

\subsection{Run-time Checks} \label{sec:empirical-runtime}

\begin{table}
  \caption{Percentage of null-dereference checks which Graduator found to be
    redundant.}
  \begin{small}
    \begin{tabular}{lrrr}
      repository & dereference sites & eliminated checks & percent eliminated \\
      \hline
      keyvaluestore & 419 & 156 & 37\% \\
      uLeak & 620 & 241 & 39\% \\
      butterknife & 2773 & 1129 & 41\% \\
      jib & 5896 & 2499 & 42\% \\
      skaffold-tools-for-java & 366 & 185 & 51\% \\
      picasso & 2719 & 1458 & 54\% \\
      meal-planner & 858 & 475 & 55\% \\
      caffeine & 9455 & 5701 & 60\% \\
      AutoDispose & 3218 & 1993 & 62\% \\
      ColdSnap & 6360 & 4325 & 68\% \\
      ReactiveNetwork & 2097 & 1626 & 78\% \\
      okbuck & 19089 & 15130 & 79\% \\
      FloatingActionButtonSpeedDial & 3049 & 2581 & 85\% \\
      QRContact & 1272 & 1171 & 92\% \\
      OANDAFX & 2216 & 2056 & 93\% \\
      \hline
      overall & 60407 & 40726 & 67\%
    \end{tabular}
  \end{small}
  \label{tab:dereferences}
\end{table}

For the same set of 15 repositories analyzed by \textsc{NullAway}, we performed another experiment using our prototype. We configured Graduator to ignore \emph{all} annotations, so in effect, every field, argument, and return value was annotated as \qm. For each repository, we counted all the locations where Graduator gave a \code{GRADUAL\_STATIC}, \code{GRADUAL\_CHECK}, or \code{GRADUAL\_BOUNDARY} warning, and compared that number to the total number of pointer dereferences in the code. By ignoring annotations, we ensured that each of these warnings appeared on dereferences, rather than allowing early checks at, e.g., method boundaries. We also ran analogous experiments with annotations enabled, but the number of run-time check warnings found were very similar to the numbers found with annotations disabled.

Table~\ref{tab:dereferences} shows what percentage of these dereference sites received no static warnings or run-time checks. Recall that Java automatically checks all dereferences to ensure that they are not null. Because GNPA is sound, this figure shows the percentage of null checks that are provably redundant, and could be safely removed by an ahead-of-time compiler.

Since we were able to eliminate an average of $67\%$ of the null checks which Java automatically inserts, this experiment suggests the answer to RQ\ref{q:eliminate-checks} is yes. Note that these numbers only discuss the number of dereferences that appear in the code, and do not take into account which of these dereferences are executed more or less frequently at run-time.

This also illustrates an important practical difference between GNPA and other null-pointer analyses. While a sound static analysis can be used to prove the redundancy of run-time checks, and an unsound static analysis can be used to reduce the number of false positives, neither of those can do both at the same time. On the other hand, a gradual analysis can both prove the redundancy of run-time checks and reduce reported false positives.

\section{Related Work}\label{sec:related}

As discussed previously, our work builds on prior research in gradual typing: the criteria for gradual type systems~\cite{siek2015refined} and the Abstracting Gradual Typing methodology, which develops a gradual type system from a purely static one~\cite{garcia2016abstracting}. In contrast to prior work in gradual typing, we address the challenges of tracking transitive dataflow relationships, rather than the local checks of typical type systems. In doing so, we gradualize, for the first time, the abstract interpretation of a program~\cite{cousot:popl1977}, and the canonical dataflow analysis fix-point algorithm~\cite{kildall1973unified}.

The most closely related work in program analysis consists of \textit{hybrid analyses}, which combine static and dynamic analysis techniques to counteract the weaknesses inherent to each approach. For example, Choi \etal~\cite{10.1145/512529.512560} used a static analysis to substantially lower the run-time overhead of a dynamic data race analysis. Prior work on hybrid program analyses combines static and dynamic techniques in ad-hoc ways. Instead, we propose a principled methodology for deriving a hybrid (gradual) analysis from a static one, and show that the resulting analysis adheres to desirable properties such as soundness and the gradual guarantee.

There is a large body of literature on static program analysis, including multiple specialized conferences. Our work opens the door to gradual versions of them. Previously, we discussed existing null-pointer analysis tools~\cite{FacebookInfer}, \cite{banerjee2019nullaway} and frameworks~\cite{papi2008practical}, and how GNPA is an improvement over them. Notably, our prototype is implemented in Infer's framework~\cite{FacebookInfer}.

The Granullar type system \cite{brotherston2017granullar} and the Blame for Null calculus\cite{blameForNull} are gradual type systems for nullness, and thus solve a related problem to GNPA.
The main difference in our work is that we use dataflow analysis instead of typing.
This results in a significantly different user experience, as a full static specification within a gradual type system typically requires many more types to be specified (e.g. on all local variables) compared to a dataflow analysis, where for example we do not require (or even allow) nullity annotations on local variables.
Basing our work on dataflow analysis also has a major impact on the technical development, requiring the novel lattice-based gradualization framework described in this paper rather than the well-known type-based gradualization approaches used in Granullar and Blame for Null.
Blame for Null also investigates the notion of blame, which we leave for future work in the program analysis setting.

Contract checking~\cite{eiffelBook, findlerFelleisen:icfp2002} can be used to check properties like nullness. Building on the idea of hybrid type checking~\cite{knowles2010hybrid}, Xu \etal~\cite{xu2012hybrid} explored how to check contracts using a hybrid of static and dynamic analysis.  Their work was specialized to the context of logical assertions, whereas we are in the area of lattice-based program analyses.  It is also unclear whether their approach conforms to the gradual guarantee.

O'Hearn \etal~\cite{o2019incorrectness} proposed Incorrectness Logic as a means of proving that a program has a bug, rather than proving it correct.  This is consistent with our goal of reducing false positives, but it stays in the realm of static reasoning, and therefore gives up soundness.  In contrast, we reduce false positives without giving up soundness by adding run-time checks.

\section{Conclusion}\label{sec:conclusion}
This paper is the first work on gradual program analysis. We introduced a framework which transforms abstract interpretation based static analyses relying on annotations into gradual ones. Gradual analyses handle missing annotations specially, allowing them to smoothly leverage both static and dynamic techniques. Static information is used where possible and dynamic information where necessary to reduce false positives while preserving soundness. Such analyses are also \emph{conservative extensions} of their underlying static analyses and adhere to \emph{gradual guarantees}, which state that losing precision is harmless. When presenting our framework, we developed a gradual null-pointer analysis, GNPA, with the previously mentioned properties that reduces false positives compared to some popular existing tools.

Importantly, the gradual framework can be applied as described to any abstract interpretation based static analysis under the following restrictions. The analysis should support annotations, have a finite-height semilattice, a monotonic, locally-sound flow function, a safety function, and operate on a first-order, procedural, imperative programming language. Additionally, checking membership in the semilattice should be decidable. Thus, initial followup work could include gradual taint analysis, to which our framework immediately applies. Finally, we do not support widening, but we do support context-sensitivity. In the future, we plan to explore extensions of our framework for infinite-height semilattices and widening; this would allow gradualization of other analyses, such as interval analysis. Still further work could include, for instance, pointer analyses, which do not have analogues in the field of gradual typing.

On the empirical side, there are further research questions to be answered: How often does a gradual analysis catch bugs statically versus how often does it catch them at run time? Is performance lost or gained when run time checks are inserted earlier via annotations rather than just-in-time? Finally, a gradual analysis will still report false positives anywhere its base static analysis is utilized and reports false positives. As a result, we plan to explore the aforementioned research questions, including the trade-off between gradual analyses reducing false positives and being conservative extensions of underlying static analyses.



\bibliography{references}

\appendix
\section{Appendix}

\subsection{Proofs}

These proofs apply generally to any particular language/semilattice/analysis that fits within the bounds of our formal framework, of which the GNPA formalism detailed in the paper is just a particular example. We left out a few formal details in the main body of the paper, for presentation's sake; we now formalize those missing details, before proceeding to the proofs.

\begin{itemize}
\item Our case study language declares programs $p \in \sProg$ to satisfy the following well-formedness rules:
  \begin{enumerate}
    \item \textit{Unique entry point to the program:} There exists exactly one
          node $v_0 \in \sVert_p$ such that $\fInst_p(v_0) = (\cMain)$. This
          node has no predecessors and serves as the entry point to $p$.
    \item \textit{Every node belongs to exactly one procedure, or to \cMain:}
          Let $\fDescend : \sVert_p \to \powerset^{+}(\sVert_p)$ give the
          descendants of each node in the control flow graph. The set
          $\{\fDescend(v_0)\} \union \set{\fDescend(\fProc(m))}{m \in \sProc}$
          is a partition of $\sVert_p$.
    \item \textit{Always a path to return from a procedure:} For each $u \in
          \sVert_p$ there exists at least one node $\inst[v]{\stmtReturn{y\code{@}a}}
          \in \fDescend(u)$. If $v \in \fDescend(\stmtProc{m\code{@}a'}{y\code{@}b})$
          then each such $v$ must have $a = a'$.
    \item \textit{Call sites agree with procedure annotations:} For each
          $\inst{\stmtCall{x}{m\code{@}a}{y\code{@}b}}$, the annotations must match the
          procedure signature $\fProc(m) = \stmtProc{m\code{@}a}{y'\code{@}b}$.
    \item For every $\inst[u]{\iota} \in \sVert_p$:
          \begin{enumerate}
            \item \textit{Always a branch to follow:} If $\iota = \stmtBranch{y}$
                  then $u$ has exactly two successors $\inst{\stmtIf{y}}$ and
                  $\inst{\stmtElse{y}}$.
            \item \textit{No dead code after return:} If $\iota = \stmtReturn{y\code{@}a}$ then $u$ has no successors.
            \item \textit{Control flow is unique:} Otherwise $u$ has exactly one
                  successor that is not an if or else node.
          \end{enumerate}
  \end{enumerate}
\item The property our safety function must satisfy is that given a state $\xi =
  \stat{\fram{\rho}{\inst[v]{\iota}} \cdot E \cdot S}{\mu}$, if
\[
  \pDesc(\rho, \set{x \mapsto \fSafe\bb{\iota}(x)}{x \in \sVar})
\]
then $\xi \stepsto_p \xi'$ for some $\xi' \in \sState_p$. Also, these safe
values must come directly from the annotations.
\item For any $\Coll{a} \in \powerset^{+}(\sAbst)$ and $\Grad{b} \in
\Grad{\sAbst}$,
\begin{enumerate}
  \item $\Coll{a} \subseteq \gamma(\alpha(\Coll{a}))$ (``soundness''), and
  \item $\Coll{a} \subseteq \gamma(\Grad{b})$ implies $\alpha(\Coll{a})
    \preciser \Grad{b}$ (``optimality'').
\end{enumerate}
\item The associativity example in Section~\ref{sec:gradual-join} shows that in some cases we need to make $\Grad{\sAbst}$ a
\emph{strict} superset of $\{\aNa,~\aNu,~\aNN,~\qm\}$, in order for $\Gradjoin$ to be associative.
One approach could be to define $\Grad{\sAbst}$ to have an element for
\emph{every} subsemilattice of $\sAbst$; we will call this the ``full lifting''
of $\sAbst$. It can be shown that $\alpha$ always exists for the full lifting,
and that $\Gradjoin$ is always associative in the full lifting. Unfortunately,
even if the height of $\sAbst$ is finite, the height of the full lifting is not
necessarily finite; that is, if $\Grad{\sAbst}$ is the full lifting then there
can exist sequences $\Grad{a}_1, \Grad{a}_2, \ldots \in \Grad{\sAbst}$ such that
$\Grad{a}_k \Gradjoin \Grad{a}_{k + 1} = \Grad{a}_{k + 1}$ for all $k$.

To address this, we will treat the full lifting as a sort of ``universe,''
consider $\{\aNa,~\aNu,$ $\aNN,~\qm\}$ to be a generating set, and let $\Grad{\sAbst}$ be the subset
of the full lifting generated by $\{\aNa,~\aNu,~\aNN,~\qm\}$ under the operation $\Gradjoin$. We
show in subsection~\ref{subsec:semilattice-properties} that this is equivalent
to saying
\[
  \Grad{\sAbst} = \sAbst \union \{\qm\} \union \set{\qm[a]}{a \in \sAbst}
  \qq{where}
  \gamma(\qm[a]) = \set{b \in \sAbst}{a \under b}.
\]
We will call this the ``small lifting'' of $\sAbst$, and it is the lifting we
will use to construct gradual analyses. The abstraction function $\alpha$ always
exists on the small lifting $\Grad{\sAbst}$, and $(\Grad{\sAbst}, \Gradjoin)$ is
a finite-height semilattice; see subsection~\ref{subsec:semilattice-properties}.

\item We insist that it is always possible to annotate a program in a way
that does not restrict its semantics. That is, for any program $p \in \sProg$,
there must exist a program $p' \in \sProg'$ such that $p'$ is the same as $p$
except for replacing every instance of $\qm$ with $\top$ (a stronger condition
than $p' \preciser p$), and such that $\sState_{p'} = \sState_p$ and the
semantics of $p'$ are equal to the semantics of $p$.
\end{itemize}

\textbf{Proposition \ref{prop:static-progress}:}
\begin{proof}
  Let $\pi = \fKildall(\fFlow, \join, p)$. Then let
  $\fram{\rho}{\inst[v]{\iota}} = E_1$ and $\sigma = \pi(v)$. Let $x \in \sVar$
  such that $\rho(x) = d \in \sVal$. Because $\xi$ is valid, $\rho(x) \in
  \fConc(\sigma(x))$. Because $p$ is valid, $\sigma(x) \under
  \fSafe\bb{\iota}(x)$, so $\rho(x) \in \fConc(\fSafe\bb{\iota}(x))$. Finally,
  $x$ was arbitrary, so by the property of the safety function, $\xi \stepsto_p
  \xi'$ for some $\xi' \in \sState_p$.
\end{proof}

\begin{lemma}\label{lem:fixpoint}
  Let $(A, \join)$ be a semilattice (whose join function induces the partial
  order $\under$), let $\fFlow : \sInst \times \sMap_A \pto \sMap_A$ (where
  $\sMap_A = \sVar \pto A$) be monotonic in the second parameter, and let $p \in
  \sProg$. If $\pi = \fKildall(\fFlow, \join, p)$ and
  $\arcsto{p}{\inst[v_1]{\iota}}{v_2}$ then $\fFlow\bb{\iota}(\pi(v_1)) \under
  \pi(v_2)$.
\end{lemma}
\begin{proof}
  We proceed by showing that the following is a loop invariant for the
  \textbf{while} loop in lines \ref{line:while}--\ref{line:while-end} of
  Algorithm~\ref{alg:kildall}: if $\arcsto{p}{\inst[v_1]{\iota}}{v_2}$ and
  $\fFlow\bb{\iota}(\pi(v_1)) \not\under \pi(v_2)$, then $v_1 \in V$. On the
  first iteration, the invariant clearly holds because $V = \sVert_p$. Now,
  assume that the invariant holds at the beginning of an iteration. We show that
  the following is a loop invariant for the \textbf{for} loop in lines
  \ref{line:for}--\ref{line:for-end}: if $U$ is the set of all $u$ that we have
  not reached yet, then all violations of the outer invariant have $v_1 = v$ and
  $v_2 \in U$. This holds at the first iteration because the only thing we
  removed from $V$ was $v$, and $\pi$ is unchanged. Next assume that the inner
  invariant holds at the beginning of an iteration of the inner loop. The
  \textbf{if} statement in lines \ref{line:if}--\ref{line:if-end} runs iff $v,
  u$ violate the outer invariant. Because $\sigma' \under \pi(u) \join \sigma'$,
  no violation with $v_1 = v$ has $v_2 = u$ after line \ref{line:join}, although
  we may now have some violations with $v_1 = u$. But after line
  \ref{line:enqueue}, we no longer have any violations involving $u$, so all
  violations now have $v_2 \in U \setminus \{u\}$ and again $v_1 = v$. After
  this inner loop exits, we no longer have any violations of the outer invariant
  because $U = \varnothing$, so the outer invariant also holds. This completes
  the proof, because $V = \varnothing$ when the outer loop exits.
\end{proof}

\textbf{Proposition \ref{prop:static-preservation}:}
\begin{proof}
  Let $\pi = \fKildall(\fFlow, \join, p)$. Then let $\stat{S_1}{\mu_1} = \xi$
  and $\stat{S_2}{\mu_2} = \xi'$. If $S_2 = \fram{\varnothing}{v_2} \cdot S_1$
  then $\pi(v_2)$ describes $\varnothing$ vacuously. Otherwise, $S_1 = S' \cdot
  \fram{\rho_1}{v_1} \cdot S$ and $S_2 = \fram{\rho_2}{v_2} \cdot S$ where
  $\arcsto{p}{v_1}{v_2}$. Let $\sigma_1 = \pi(v_1)$ and $\sigma_2 = \pi(v_2)$.
  Because $\xi$ is valid, $\sigma_1$ describes $\rho_1$. By local soundness,
  $\sigma_2' = \fFlow\bb{\iota}(\sigma_1)$ describes $\rho_2$. Then $\sigma_2'
  \under \sigma_2$ by Lemma~\ref{lem:fixpoint} (with $A = \sAbst$), so
  $\sigma_2$ describes $\rho_2$. In each of these cases, the top stack frame of
  $S_2$ is valid. All other frames are the same as those of $S_1$, so $\xi'$ is
  valid.
\end{proof}

\begin{proposition}\label{prop:lifting-small}
  $\Grad{\sAbst}$ is the subset of the full lifting generated by $\sAnn$ via
  $\Gradjoin$.
\end{proposition}
\begin{proof}
  Let $(\Grad{\sAbst}', \Gradjoin)$ be the full lifting of $\sAbst$ with the
  corresponding lifted join function, and let
  \[
    \Grad{\sAbst}
    = \sAbst \union \{\qm\} \union \set{\qm[a]}{a \in \sAbst}
    \subseteq \Grad{\sAbst}'
  \]
  be the small lifting. First note that $a \Gradjoin \qm = \qm[a]$ for all $a
  \in \sAbst$, so $\Grad{\sAbst}$ is a subset of the set generated by $\sAnn$
  via $\Gradjoin$. Then for $\Grad{a}, \Grad{b} \in \Grad{\sAbst}$,
  \[
    \Grad{a} \Gradjoin \Grad{b} = \begin{cases}
      a \join b & \text{if }\Grad{a} = a \in \sAbst\text{ and }\Grad{b} = b \in \sAbst \\
      \qm[a] & \text{if }\Grad{a} = a \in \sAbst\text{ and }\Grad{b} = \qm \\
      \qm[(a \join b)] & \text{if }\Grad{a} = a \in \sAbst\text{ and }\Grad{b} = \qm[b]\text{ for some }b \in \sAbst \\
      \qm & \text{if }\Grad{a} = \qm\text{ and }\Grad{b} =\qm \\
      \qm[b] & \text{if }\Grad{a} = \qm\text{ and }\Grad{b} = \qm[b]\text{ for some }b \in \sAbst \\
      \qm[(a \join b)] & \text{if }\Grad{a} = \qm[a]\text{ for some } \in \sAbst\text{ and }\Grad{b} = \qm[b]\text{ for some } \in \sAbst \\
      \Grad{b} \Gradjoin \Grad{a} & \text{otherwise}
    \end{cases}
  \]
  so $\set{\Grad{a} \Gradjoin \Grad{b}}{\Grad{a}, \Grad{b} \in \Grad{\sAbst}}
  \subseteq \Grad{\sAbst}$. Thus, $\Grad{\sAbst}$ is equal to the set generated
  by $\sAnn$ via $\Gradjoin$.
\end{proof}

\begin{proposition}
  $\Grad{\sAbst}$ has an abstraction function $\alpha$.
\end{proposition}
\begin{proof}
  Let $\Coll{a} \in \powerset^{+}(\sAbst)$, and let $A = \set{\Grad{b} \in
  \Grad{\sAbst} \setminus \{\qm\}}{\Coll{a} \subseteq \gamma(\Grad{b})}$. If any
  such $\gamma(\Grad{b})$ is a singleton then $\alpha(\Coll{a}) = \Grad{b}$ and
  we're done. If $A = \varnothing$ then $\alpha(\Coll{a}) = \qm$. Now without
  loss of generality, we assume that each of those $\Grad{b}$ elements is of the
  form $\qm[b]$ for some $b \in \sAbst$; that is, there exists an injective
  ``root'' map $r : A \to \sAbst$ given by $r(\qm[b]) = b$. Let $A_0 = r(A)$.

  Next we inductively define an ascending chain $b_k$ along with a sequence of
  sets $A_k$ for $k \in \Nat$; our base case is $A_0$. Choose $b_k \in A_k$ and let
  \[
    A_{k + 1} = \set{b \in A_k}{b \join b_k \neq b_k}.
  \]
  If $A_{k + 1} = \varnothing$ then we end the chain. Otherwise, choose $b_k'
  \in A_{k + 1}$ and let $b_{k + 1} = b_k \join b_k'$. By the construction of
  $A_{k + 1}$, we know that $b_{k + 1} \neq b_k$, so we have continued our
  ascending chain to be $b_0 \strictunder \cdots \strictunder b_k \strictunder
  b_{k + 1}$ because
  \[
    b_k \join b_{k + 1}
    = b_k \join (b_k \join b_k')
    = (b_k \join b_k) \join b_k'
    = b_k \join b_k'
    = b_{k + 1}.
  \]

  Let $h$ be the height of $\sAbst$, so we know that our chain has height $n
  \leq h$. By construction, for every $b \in A_0$ we have $b \join b_k = b_k$
  for some $0 \leq k \leq n$, which means that $\gamma(\qm[b]) \supseteq
  \gamma(\qm[b_k])$. Given that $\gamma(\qm[b_0]) \supseteq \cdots \supseteq
  \gamma(\qm[b_n])$, we see that $\gamma(\qm[b_n]) = \Intersect \gamma(A)$, so
  we can define $\alpha(\Coll{a}) = \qm[b_n]$.
\end{proof}

\textbf{Proposition \ref{prop:semilattice}:}
\begin{proof}
  We have already shown that $\Gradjoin$ is commutative and idempotent, so it
  only remains to show that $\Gradjoin$ is associative. But associativity
  follows immediately from the proof of Proposition~\ref{prop:lifting-small}.
\end{proof}

\textbf{Proposition \ref{prop:finite-height}:}
\begin{proof}
  In this proof, we write $\Grad{a} \under \Grad{b}$ to mean $\Grad{a} \Gradjoin
  \Grad{b} = \Grad{b}$ for $\Grad{a}, \Grad{b} \in \Grad{\sAbst}$, and also
  write $\Grad{a} \strictunder \Grad{b}$ to mean $\Grad{a} \under \Grad{b}$ and
  $\Grad{a} \neq \Grad{b}$. Note that these are not the same as the lifted
  relation $\Gradunder$, although $\Gradunder$ and this definition of $\under$
  both coincide when restricted to $\sAbst \times \sAbst$.

  By the definition of height, there exists a (not necessarily unique) longest
  ascending chain $a_0 \strictunder \cdots \strictunder a_n$ in $\sAbst$. Since
  $n > 0$ we know that $\gamma(\qm[a_{n - 1}])$ is not a singleton because $a_{n
  - 1}, a_n \in \gamma(\qm[a_{n - 1}])$. Thus, $\qm[a_{n - 1}] \neq a_{n - 1}$.
  We can then calculate
  \begin{gather*}
    a_{n - 1} \Gradjoin \qm[a_{n - 1}] = \qm[(a_{n - 1} \join a_{n - 1})] = \qm[a_{n - 1}], \\
    \qm[a_{n - 1}] \Gradjoin \qm[a_n] = \qm[(a_{n - 1} \join a_n)] = \qm[a_n],
  \end{gather*}
  so $a_{n - 1} \strictunder \qm[a_{n - 1}] \strictunder \qm[a_n]$ because $a_{n
  - 1} \neq a_n$ implies $\qm[a_{n - 1}] \neq \qm[a_n]$. This shows that the
  height of the small lifting is at least $n + 1$.

  Now assume that there exists an ascending chain $\Grad{a}_0 \strictunder
  \cdots \strictunder \Grad{a}_{n + 2}$ in $\Grad{\sAbst}$. Note that for $k >
  0$, if $\Grad{a}_k = \qm$ then $\Grad{a}_{k - 1} \Gradjoin \qm = \qm$, which
  implies $\Grad{a}_{k - 1} = \bot$, so $\Grad{a}_k = \qm[\bot]$. Thus for $k >
  0$ either $\Grad{a}_k = a_k$ or $\Grad{a}_k = \qm[a_k]$, allowing us to define
  a new chain $a_1 \under \cdots \under a_{n + 2}$. If $\Grad{a}_0 = \qm$ then
  we must have $\Grad{a}_1 = \qm[a_1] \neq a_1$, because $a_1, a_2 \in
  \gamma(\qm[a_1])$. In this case we can replace $\Grad{a}_0$ with $a_1$, so
  without loss of generality we can assume that no element of the chain is
  $\qm$. Next, if $\Grad{a}_k = \qm[a_k]$ for some $0 \leq k < n + 2$, we can
  use $\Grad{a}_k \under \Grad{a}_{k + 1}$ to see that $\Grad{a}_{k + 1} =
  \Grad{a}_k \Gradjoin \Grad{a}_{k + 1} = \qm[a_k] \Gradjoin \Grad{a}_{k + 1} =
  \qm[a_{k + 1}]$. By induction this means that if $\Grad{a}_k = a_k$ and
  $\Grad{a}_{k + 1} = \qm[a_{k + 1}]$ for some $k$, we must have $\Grad{a}_i =
  a_i$ for all $i \leq k$ and $\Grad{a}_j = \qm[a_j]$ for all $j > k$. In other
  words, we have a chain
  \[
    x_0
    \strictunder \cdots
    \strictunder x_k
    \under x_{k + 1}
    \strictunder \cdots
    \strictunder x_{n + 2}
  \]
  implying that $\sAbst$ is at least height $n + 1$, contrary to our earlier
  assumption. Thus the height of the small lifting is at most $n + 1$.
\end{proof}

\textbf{Proposition \ref{prop:conservative-static-extension}:}
\begin{proof}
  For any $a, b \in \sAbst$ we have $\gamma(a) = \{a\}$ and $\gamma(b) = \{b\}$,
  so $a \Gradjoin b = \alpha(\{a \join b\}) = a \join b$ because $\gamma(a \join
  b) = \{a \join b\}$. Thus $\Gradjoin$ is a conservative extension of $\join$.
  Similarly $\Grad{\fFlow}\bb{\iota}(\sigma) = \fFlow\bb{\iota}(\sigma)$ for
  $\iota \in \sInst'$ and $\sigma \in \sMap$, so $\Grad{\fFlow}$ is a
  conservative extension of $\fFlow$. Because $\pi = \fKildall(\fFlow, \join,
  p)$ is well-defined, it follows that $\fKildall(\Grad{\fFlow}, \Gradjoin, p) =
  \pi$.
\end{proof}

\textbf{Proposition \ref{prop:conservative-dynamic-extension}:}
\begin{proof}
  The predicate $\Gradunder$ is a conservative extension of $\under$, and the
  function $\Grad{\fSafe}$ is a conservative extension of $\fSafe$, so $p$ is
  statically valid according to the gradual analysis as well as valid according
  to the static analysis. If $\xi_1 \Gradstepsto_p \xi_2$ then trivially $\xi_1
  \stepsto_p \xi_2$ because $\xi_2 \neq \cError$. Conversely, assume that $\xi_1
  \stepsto_p \xi_2$. Let $\pi = \fKildall(\fFlow, \join, p)$. Since $p$ and
  $\xi_1$ are valid, if $\xi_1 = \stat{\fram{\rho}{\inst{\iota}} \cdot S}{\mu}$
  then $\pDesc(\rho, \set{x \mapsto \fSafe\bb{\iota}(x)}{x \in \sVar})$ by the
  same reasoning used in the proof of Proposition~\ref{prop:static-progress}.
  Then $\xi_1$ does not step to $\cError$ because $\Grad{\pDesc}$ and
  $\Grad{\fSafe}$ are conservative extensions of $\pDesc$ and $\fSafe$
  respectively. Thus, $\xi_1 \Gradstepsto_p \xi_2$.
\end{proof}

\begin{lemma}\label{lem:safe-monotonic}
  If $\iota_1, \iota_2 \in \sInst$ and $\iota_1 \preciser \iota_2$, then
  $\gamma(\Grad{\fSafe}\bb{\iota_1}(x)) \subseteq
  \gamma(\Grad{\fSafe}\bb{\iota_2}(x))$ for all $x \in \sVar$.
\end{lemma}
\begin{proof}
  Let $x \in \sVar$. If $\Grad{\fSafe}\bb{\iota_1}(x) =
  \Grad{\fSafe}\bb{\iota_2}(x)$ then the claim clearly holds. Otherwise, since
  $\iota_1$ and $\iota_2$ only differ in annotations, there must exist
  $\iota_1', \iota_2' \in \sInst'$ such that $\fSafe\bb{\iota_1'}(x) \neq
  \fSafe\bb{\iota_2'}(x)$. Therefore we know that $\Grad{\fSafe}\bb{\iota_1}(x)$
  and $\Grad{\fSafe}\bb{\iota_2}(x)$ come from corresponding operands of
  $\iota_1$ and $\iota_2$ respectively. Since $\iota_1 \preciser \iota_2$, that
  operand must be $\Grad{\fSafe}\bb{\iota_2}(x) = \qm$ for $\iota_2$ in order
  for the safety values to be different. Thus we have
  $\gamma(\Grad{\fSafe}\bb{\iota_1}(x)) \subseteq \gamma(\qm) =
  \gamma(\Grad{\fSafe}\bb{\iota_2}(x))$.
\end{proof}

\begin{lemma}\label{lem:safe-gradual}
  Let $p \in \sProg$ and $\xi = \stat{\fram{\rho}{\inst[v]{\iota}} \cdot E \cdot
  S}{\mu} \in \sState_p$. If $\Grad{\pDesc}(\rho, \set{x \mapsto
  \Grad{\fSafe}\bb{\iota}(x)}{x \in \sVar})$ then $\xi \stepsto_p \xi'$ for some
  $\xi' \in \sState_p$.
\end{lemma}
\begin{proof}
  We know there exists a program $p' \in \sProg'$ more precise than $p$ whose
  states and semantics are the same as those of $p$, so in particular $\xi \in
  \sState_{p'}$, but $\iota' = \fInst_{p'}(v)$ is not necessarily equal to
  $\iota$ since all instances of $\qm$ in $p$ have been replaced with $\top$ in
  $p'$. Next, by the definition of $\Grad{\pDesc}$ there exists some $\sigma \in
  \sMap$ such that $\pDesc(\rho, \sigma)$ and $\sigma(x) \in
  \gamma(\Grad{\fSafe}\bb{\iota}(x))$ for all $x \in \sVar$. Now let $x \in
  \sVar$. If $\Grad{\fSafe}\bb{\iota}(x) = a \in \sAbst$ then
  $\fSafe\bb{\iota'}(x) = \sigma(x)$, so $\rho(x) \in \fConc(a)$. Otherwise
  there exist $\iota_1, \iota_2 \in \sInst'$ such that $\fSafe\bb{\iota_1}(x)
  \neq \fSafe\bb{\iota_2}(x)$, so we know that $\fSafe\bb{\iota'}(x)$ is an
  operand of $\iota'$. But the corresponding operand of $\iota$ must be $\qm$
  since otherwise we would not have multiple values in
  $\gamma(\Grad{\fSafe}\bb{\iota}(x))$, so we have $\fSafe\bb{\iota'}(x) = \top$
  and trivially $\rho(x) \in \fConc(\top)$. Thus, $\pDesc(\rho, \set{x \mapsto
  \fSafe\bb{\iota'}(x)}{x \in \sVar}))$, so $\xi \stepsto_p \xi'$ for some $\xi'
  \in \sState_{p'} = \sState_p$, which means $\xi \stepsto_p \xi'$.
\end{proof}

\textbf{Proposition \ref{prop:gradual-progress}:}
\begin{proof}
  Let $\fram{\rho}{\inst{\iota}} = E_1$. If $\lnot\Grad{\pDesc}(\rho, \set{x
  \mapsto \Grad{\fSafe}\bb{\iota}(x)}{x \in \sVar})$ then $\xi \Gradstepsto_p
  \cError$. Otherwise, $\xi \stepsto_p \xi'$ for some $\xi' \in \sState_p
  \subset \Grad{\sState}_p$ by Lemma~\ref{lem:safe-gradual}, so $\xi
  \Gradstepsto_p \xi'$ because $\xi$ does not step to $\cError$.
\end{proof}

\begin{lemma}\label{lem:flow-gradual}
  Let $p \in \sProg$ and $\Grad{\sigma} \in \Grad{\sMap}$, and let $\xi =
  \stat{S' \cdot \fram{\rho}{\inst[v]{\iota}} \cdot S}{\mu}$ and $\xi' =
  \stat{\fram{\rho'}{v'} \cdot S}{\mu}$. If $\xi \stepsto_p \xi'$ and
  $\Grad{\pDesc}(\rho, \Grad{\sigma})$, then $\Grad{\pDesc}(\rho',
  \Grad{\fFlow}\bb{\iota}(\Grad{\sigma}))$.
\end{lemma}
\begin{proof}
  We know there exists a program $p' \in \sProg'$ more precise than $p$ whose
  states and semantics are the same as those of $p$, so in particular $\xi, \xi'
  \in \sState_{p'}$, and $\xi \stepsto_{p'} \xi'$. However, $\iota' =
  \fInst_{p'}(v)$ is not necessarily equal to $\iota$ since all instances of
  $\qm$ in $p$ have been replaced with $\top$ in $p'$. Next, by the definition
  of $\Grad{\pDesc}$ there exists some $\sigma \in \sMap$ such that
  $\pDesc(\rho, \sigma)$ and $\sigma(x) \in \gamma(\Grad{\sigma}(x))$ for all $x
  \in \dom(\Grad{\sigma})$. By local soundness, $\pDesc(\rho',
  \fFlow\bb{\iota'}(\sigma))$. But by the definition of $\Grad{\fFlow}$ we know
  $(\fFlow\bb{\iota'}(\sigma))(x) \in
  \gamma((\Grad{\fFlow}\bb{\iota}(\Grad{\sigma}))(x))$ for all $x \in
  \dom(\fFlow\bb{\iota'}(\sigma))$, so $\Grad{\pDesc}(\rho',
  \Grad{\fFlow}\bb{\iota}(\Grad{\sigma}))$.
\end{proof}

\textbf{Proposition \ref{prop:gradual-preservation}:}
\begin{proof}
  Let $\Grad{\pi} = \fKildall(\Grad{\fFlow}, \Gradjoin, p)$. Then let
  $\stat{S_1}{\mu_1} = \xi$ and $\stat{S_2}{\mu_2} = \xi'$. If $S_2 =
  \fram{\varnothing}{v_2} \cdot S_1$ then $\Grad{\pi}(v_2)$ describes
  $\varnothing$ vacuously. Otherwise, $S_1 = S' \cdot \fram{\rho_1}{v_1} \cdot
  S$ and $S_2 = \fram{\rho_2}{v_2} \cdot S$ where $\arcsto{p}{v_1}{v_2}$. Let
  $\Grad{\sigma}_1 = \Grad{\pi}(v_1)$ and $\Grad{\sigma}_2 = \Grad{\pi}(v_2)$.
  Because $\xi$ is valid, $\Grad{\sigma}_1$ describes $\rho_1$. By
  Lemma~\ref{lem:flow-gradual}, $\Grad{\sigma}_2' =
  \Grad{\fFlow}\bb{\iota}(\Grad{\sigma}_1)$ describes $\rho_2$. Then
  $\Grad{\sigma}_2' \Gradjoin \Grad{\sigma}_2 = \Grad{\sigma}_2$ by
  Lemma~\ref{lem:fixpoint} (with $A = \Grad{\sAbst}$, $\join = \Gradjoin$, and
  $\fFlow = \Grad{\fFlow}$), so $\Grad{\sigma}_2$ describes $\rho_2$. In each of
  these cases, the top stack frame of $S_2$ is valid. All other frames are the
  same as those of $S_1$, so $\xi'$ is valid.
\end{proof}

\begin{lemma}\label{lem:flow-monotonic}
  Let $\iota_1, \iota_2 \in \sInst$ such that $\iota_1 \preciser \iota_2$.\\
  Then $\gamma((\Grad{\fFlow}\bb{\iota_1}(\Grad{\sigma}))(x)) \subseteq
  \gamma((\Grad{\fFlow}\bb{\iota_2}(\Grad{\sigma}))(x))$ for all $\Grad{\sigma}
  \in \Grad{\sMap}$ and $x \in \sVar$.
\end{lemma}
\begin{proof}
  Using the notation from the definition of $\Grad{\fFlow}$, we have $I_1
  \subseteq I_2$, so the lemma holds by the properties of $\alpha$.
\end{proof}

\begin{lemma}\label{lem:fixpoint-gradual}
  Let $p_1, p_2 \in \sProg$ such that $p_1 \preciser p_2$. Let $\pi_1 =
  \fKildall(\Grad{\fFlow}, \Gradjoin, p_1)$ and $\pi_2 =
  \fKildall(\Grad{\fFlow}, \Gradjoin, p_2)$. Let $v \in \sVert_{p_1} =
  \sVert_{p_2}$. Let $\sigma_1 = \pi_1(v)$ and $\sigma_2 = \pi_2(v)$. Then
  $\gamma(\sigma_1(x)) \subseteq \gamma(\sigma_2(x))$ for all $x \in
  \dom(\sigma_1)$.
\end{lemma}
\begin{proof}
  We proceed by running Algorithm~\ref{alg:kildall} in parallel for $p_1$ and
  $p_2$ and showing that the lemma statement is a loop invariant for the
  \textbf{while} loop in lines \ref{line:while}--\ref{line:while-end}. On the
  first iteration, the invariant clearly holds because $\dom(\Grad{\sigma}_1) =
  \varnothing$. Now, assume that the invariant holds at the beginning of an
  iteration. Without loss of generality we can assume $v$ to be chosen to be the
  same for both sides, because if $v_1 \notin V_2$ or $v_2 \notin V_1$ then the
  \textbf{if} statement on line \ref{line:if} will never run for the first or
  second side, respectively. After line \ref{line:sigma} we have
  $\gamma(\sigma_1(x)) \subseteq \gamma(\sigma_2(x))$ for all $x$ by assumption.
  Then after line \ref{line:flow} we have $\gamma(\sigma'_1(x)) \subseteq
  \gamma(\sigma'_2(x))$ for all $x$ by Lemma~\ref{lem:flow-monotonic}. The in
  the inner \textbf{for} loop, we enter the \textbf{if} statement in line
  \ref{line:if} exactly when the assignment statement on line \ref{line:join}
  would have an effect. By the properties of $\Gradjoin$, the invariant still
  holds for $\pi_1(u)$ and $\pi_2(u)$ after line \ref{line:join}. This accounts
  for all the elements of $\pi_1$ and $\pi_2$ that we change. We have thus
  completed the proof.
\end{proof}

\textbf{Proposition \ref{prop:static-gradual-guarantee}:}
\begin{proof}~\\
  Let $\Grad{\pi}_1 = \fKildall(\Grad{\fFlow}, \Gradjoin, p_1)$ and
  $\Grad{\pi}_2 = \fKildall(\Grad{\fFlow}, \Gradjoin, p_2)$. Let $v \in
  \sVert_{p_1} = \sVert_{p_2}$ and $x \in \sVar$, let $\iota_1 =
  \fInst_{p_1}(v)$ and $\iota_2 = \fInst_{p_2}(v)$, and let $\Grad{\sigma}_1 =
  \Grad{\pi}_1(v)$ and $\Grad{\sigma}_2 = \Grad{\pi}_2(v)$. By
  Lemma~\ref{lem:fixpoint-gradual} we know $\gamma(\Grad{\sigma}_1(x)) \subseteq
  \gamma(\Grad{\sigma}_2(x))$. Also, by Lemma~\ref{lem:safe-monotonic} we know
  $\gamma(\Grad{\fSafe}\bb{\iota_1}(x)) \subseteq
  \gamma(\Grad{\fSafe}\bb{\iota_2}(x))$. Then by the definition of $\Gradunder$,
  if $\Grad{\sigma}_1(x) \Gradunder \Grad{\fSafe}\bb{\iota_1}(x)$ then
  $\Grad{\sigma}_2(x) \Gradunder \Grad{\fSafe}\bb{\iota_2}(x)$.
\end{proof}

\textbf{Proposition \ref{prop:dynamic-gradual-guarantee}:}
\begin{proof}
  Because $\xi_2 \neq \cError$, we know that $\xi_1 \stepsto_{p_1} \xi_2$. This
  means that $\xi_1 \stepsto_{p_2} \xi_2$ because $\xi_2$ is the same as $\xi_1$
  except with possibly some annotations removed. Thus, it only remains to show
  that $\xi_1$ does not step to $\cError$ under $\Gradstepsto_{p_2}$. Assume
  that $\xi = \stat{\fram{\rho}{v} \cdot S}{\mu}$ where $\fInst_{p_1}(v) =
  \iota_1$ and $\fInst_{p_2}(v) = \iota_2$. Because $\xi_1$ does not step to
  $\cError$, we know that $\Grad{\pDesc}(\rho, \set{x \mapsto
  \Grad{\fSafe}\bb{\iota_1}(x)}{x \in \sVar})$. This means that there exists
  some $\sigma \in \sMap$ such that $\pDesc(\rho, \sigma)$ and $\sigma(x) \in
  \gamma(\Grad{\fSafe}\bb{\iota_1}(x))$ for all $x \in \sVar$. By
  Lemma~\ref{lem:safe-monotonic} we know that $\sigma(x) \in
  \gamma(\Grad{\fSafe}\bb{\iota_1}(x))$ for all $x \in \sVar$. This completes
  the proof, because by the definition of $\pDesc$ we now know that
  $\Grad{\pDesc}(\rho, \set{x \mapsto \Grad{\fSafe}\bb{\iota_2}(x)}{x \in
  \sVar})$, so $\xi_1$ does not step to $\cError$ under $\Gradstepsto_{p_2}$, so
  $\xi_1 \Gradstepsto_{p_2} \xi_2$.
\end{proof}

\textbf{Proposition \ref{prop:runtime-checks}:}
\begin{proof}
  We know that $\lnot\Grad{\pDesc}(\rho, \set{x \mapsto
  \Grad{\fSafe}\bb{\iota}(x)}{x \in \sVar})$. By the definitions of
  $\Grad{\pDesc}$ and $\pDesc$, there is some $x \in \sVar$ and $b \in
  \gamma(\Grad{\fSafe}\bb{\iota}(x))$ such that $\rho(x) \notin \fConc(b)$. But
  since $\xi$ is valid, there exists some $a \in \gamma((\Grad{\pi}(v))(x))$
  such that $\rho(x) \in \fConc(a)$. Thus, $\fConc(a) \nsubseteq \fConc(b)$, so
  $a \not\under b \under \bigjoin \gamma(\Grad{\fSafe}\bb{\iota}(x))$.
\end{proof}

\end{document}